\documentclass[a4paper, USenglish, cleveref, autoref, thm-restate]{lipics-v2021}

\pdfoutput=1 

\hideLIPIcs

\usepackage{complexity}
\usepackage{xspace}
\usepackage{nicefrac}

\usepackage{comment}
\usepackage[dvipsnames]{xcolor}
\usepackage{tabularray}
\usepackage{tikz}
\usepackage{booktabs}

\usepackage[export]{adjustbox}
\usepackage{mathtools}

\usepackage[export]{adjustbox}
\usepackage{csquotes}

\usepackage[vlined]{algorithm2e}
\SetArgSty{textnormal}
\DontPrintSemicolon
\SetKwProg{Fn}{function}{}{end}
\SetKwFor{ForEach}{for each}{do}{endfch}

\newtheorem*{proposition*}{Proposition}

\newcommand{\trans}{\textsc{L-BinPack}\xspace}
\newcommand{\transStrip}{\textsc{L-StripPack}\xspace}
\newcommand{\rot}{\textsc{Ortho-BinPack}$_\text{rot}$\xspace}
\newcommand{\sorting}{\textsc{BinSorting}\xspace}

\newcommand{\peri}{\textsc{L-PeriPack}\xspace}
\newcommand{\minarea}{\textsc{L-AreaPack}\xspace}

\newcommand{\LasylPacker}{\textsc{LaSyL-Packer}\xspace}
\newcommand{\smallLPacker}{\textsc{SmallL-Packer}\xspace}

\newcommand{\new}[1]{#1}

\newcommand{\ortho}{\textsc{Ortho-BinPack}\xspace}
\newcommand{\FF}{\textsc{FirstFit}\xspace}
\newcommand{\FFL}{\textsc{DensePacker}\xspace}

\newcommand{\Lskeleton}{\textsc{L-Skel-BinPack}\xspace}
\newcommand{\Zskeleton}{\textsc{Z-Skel-BinPack}\xspace}
\newcommand{\playerA}{\textsc{\textsc{Presenter}}\xspace}
\newcommand{\playerB}{\textsc{Algorithm}\xspace}

\newcommand{\Zframe}{Z-skeleton\xspace}
\newcommand{\Zframes}{Z-skeletons\xspace}

\newcommand{\opt}{\textsc{OPT}\xspace}
\newcommand{\OPT}{\textsc{OPT}\xspace}
\newcommand{\ALG}{\ensuremath{\mathcal{A}}\xspace}

\newcommand{\Lshape}{L-shape\xspace}
\newcommand{\Lshapes}{L-shapes\xspace}

\newcommand{\Lframe}{L-skeleton\xspace}
\newcommand{\Lframes}{L-skeletons\xspace}
\newcommand{\alg}{\ensuremath{\mathcal{A}}\xspace}

\newcommand{\oneDBin}{\textsc{1D-Bin-Packing}\xspace}

\newcommand{\area}{\text{area}}

\newcommand{\lx}{\ell_x}
\newcommand{\ly}{\ell_y}
\newcommand{\wx}{w_x}
\newcommand{\wy}{w_y}

\newcommand{\SymbRect}{\tikz[scale=0.3]{\draw (0, 0) -- (0.3, 0) -- (0.3, 1) -- (0, 1) -- cycle;}}
\newcommand{\SymbSymL}{\tikz[scale=0.3]{\draw (0, 0) -- (1, 0) -- (1, 0.3) -- (0.3, 0.3) -- (0.3, 1) -- (0, 1) -- cycle;}}
\newcommand{\SymbL}{\tikz[scale=0.3]{\draw (0, 0) -- (1, 0) -- (1, 0.3) -- (0.3, 0.3) -- (0.3, 0.7) -- (0, 0.7) -- cycle;}}
\newcommand{\SymbZ}{\tikz[scale=0.3]{\draw (0, 0) -- (0.3, 0) -- (0.3, 0.35) -- (1, 0.35) -- (1, 1) -- (0.7, 1) -- (0.7, 0.65) -- (0, 0.65) -- cycle;}}

\newcommand{\SymbConv}{\tikz[scale=0.3]{\draw (0, 0) -- (1, 0) -- (0.5, 1.0) -- (0, 0.5) -- cycle;}}
\newcommand{\SymbLFrame}{\tikz[scale=0.3]{\draw[thick] (1, 0) -- (0, 0) -- (0, 1)}}
\newcommand{\SymbZFrame}{\tikz[scale=0.3]{\draw[thick] (0, 0) -- (0, 0.5) -- (1, 0.5) -- (1, 1)}}

\ExplSyntaxOn 

\cs_new:Nn \algo__subfigure:nnnnn {
    \hfill
    \begin{subfigure}[t]{#4\linewidth}
        \centering
        \includegraphics[page=#1,max~width=\linewidth]{#2}
        \subcaption{#5}
        \label{#3}
    \end{subfigure}
    \hfill\strut
}
\cs_generate_variant:Nn \algo__subfigure:nnnnn {nneee}

\seq_new:N \l_algo_subfigure_pages
\seq_new:N \l_algo_subfigure_widths
\seq_new:N \l_algo_subfigure_captions

\cs_new:Nn \algo_include_subfigures:nnnnnn {
    \seq_set_from_clist:Nn \l_algo_subfigure_pages {#3}
    \seq_set_from_clist:Nn \l_algo_subfigure_widths {#5}
    \seq_set_from_clist:Nn \l_algo_subfigure_captions {#6}
    \begin{figure}[htb]
        \centering
        \seq_map_indexed_inline:Nn \l_algo_subfigure_pages {
            \algo__subfigure:nneee
                    {##2} 
                    {#2} 
                    {fig:#1\int_to_Alph:n {##1}} 
                    {\seq_item:Nn \l_algo_subfigure_widths {##1}}
                    {\seq_item:Nn \l_algo_subfigure_captions {##1}}
        }
        \caption{#4}
        \label{fig:#1}
    \end{figure}
}
\NewDocumentCommand{\IncludeSubfigures}{mmmO{0.24,0.24,0.24,0.24}mO{,,,,,,,,,,,}}{\algo_include_subfigures:nnnnnn {#1} {#2} {#3} {#5} {#4} {#6}}

\ExplSyntaxOff

\newcommand{\where}{\mid}

\renewcommand{\epsilon}{\varepsilon}

\title{Online Packing of Orthogonal Polygons}

\nolinenumbers

\author{Tim Gerlach}%
    {Universit\"at Hamburg, Germany}%
    {tim.gerlach@uni-hamburg.de}%
    {https://orcid.org/0009-0004-6294-9235}%
    {}

\author{Benjamin Hennies}%
    {Technische Universit\"at Braunschweig, Germany}%
    {benjaminhennies@posteo.de}%
    {}%
    {}

\author{Linda Kleist}%
    {Universit\"at Hamburg, Germany}%
    {linda.kleist@uni-hamburg.de}%
    {https://orcid.org/0000-0002-3786-916X}%
    {}

\authorrunning{T. Gerlach, B. Hennies, and  L. Kleist} 

\Copyright{Tim Gerlach, Benjamin Hennies and Linda Kleist}



\acknowledgements{We thank the anonymous reviewers for their valuable comments
and suggestions, which helped to improve the quality of this work.  In
particular, we thank an anonymous reviewer for the suggestion to use
\textsc{NextFit} in the proof of \cref{thm:skeletons}(i) and improving the
analysis to obtain a competitive ratio of 2.}

\EventEditors{Hee-Kap Ahn, Michael Hoffmann, and Amir Nayyeri}
\EventNoEds{3}
\EventLongTitle{42nd International Symposium on Computational Geometry (SoCG 2026)}
\EventShortTitle{SoCG 2026}
\EventAcronym{SoCG}
\EventYear{2026}
\EventDate{June 2--5, 2026}
\EventLocation{New Brunswick, NJ, USA}
\EventLogo{socg-logo.pdf}
\SeriesVolume{367}
\ArticleNo{XX}

\begin{document}

\maketitle

\DeclareFontShape{T1}{lmr}{m}{scit} { <-> ssub * lmr/m/scsl }{}
\DeclareFontShape{T1}{lmss}{bx}{sc} { <-> ssub * lmss/bx/n }{}

\begin{abstract}
  While rectangular and box-shaped objects dominate the classic discourse of
  theoretic investigations, a fascinating frontier lies in  packing more complex
  shapes. Given recent insights that convex polygons do not allow for constant
  competitive online algorithms for diverse variants under translation, we study
  orthogonal polygons, in particular of small complexity. For translational
  packings of orthogonal 6-gons, we show that the competitive ratio of any
  online algorithm that aims to pack the items into a minimal number of unit
  bins is in $\Omega(\nicefrac{n}{\log n})$, where $n$ denotes the number of
  objects. In contrast, we show that constant competitive algorithms exist when
  the orthogonal 6-gons are symmetric or small.  For (orthogonally convex)
  orthogonal 8-gons, we show that the trivial $n$-competitive algorithm, which
  places each item in its own bin, is best-possible, i.e., every online
  algorithm has an asymptotic competitive ratio of at least $n$. This implies
  that for general orthogonal polygons, the trivial algorithm is best possible.

  Interestingly, for packing degenerate orthogonal polygons (with thickness
  $0$), called skeletons, the change in complexity is even more drastic.  While
  constant competitive algorithms for 6-skeletons exist, no online algorithm for
  8-skeletons achieves a competitive ratio better than $n$.

  For other packing variants of orthogonal 6-gons under translation, our
  insights imply the following consequences. The asymptotic competitive ratio of
  any online algorithm is in $\Omega(\nicefrac{n}{\log n})$ for strip packing,
  and there exist online algorithms with competitive ratios in $O(1)$ for
  perimeter packing,  or in $O(\sqrt{n})$ for minimizing the area of the
  bounding box. Moreover, the critical packing density is positive (if every
  object individually fits into the interior of a unit bin).
\end{abstract}

\section{Introduction}
Packing problems not only  present us with constant challenges in
everyday life, but also find applications in manufacturing
industries, logistics, and scheduling.
While rectangular and box-shaped objects dominate the classic
discourse of theoretic investigations, a fascinating frontier lies in
the packing of more complex objects, e.g., convex polygons.
The problem of packing convex polygons is particularly interesting
when restricting the allowed motions  to translations, as allowing
for rotations reduces the problem to mere rectangle packing.
Across the numerous variants -- be it strip packing, bin packing, or
perimeter packing -- the overall picture is that constant factor
approximations exist for translational packings of convex
polygons~\cite{alt_convexOffline_JoCG,alt_convexOffline_corr,onlineSortingOnlinePacking_SODA,convexOffline-improved}.
However, the landscape changes when we step into the online realm
where items arrive sequentially one after the other and decisions
have to be made irrevocably without knowing the future items: 
In this setting, no online algorithm for translational packings of convex
polygons has a constant competitive ratio, as recently shown by
Aamand, Abrahamsen, Beretta, and
Kleist~\cite{onlineSortingOnlinePacking_SODA}. This is in stark
contrast to the case of rectangles, for which constant competitive
algorithms exist for various variants,
cf.~\cref{sec:relWork}.

In this work, we study the online packing problem of orthogonal polygons which
are orthogonally convex, combining relaxed features of both orthogonality and
convexity. A polygon is \emph{orthogonally convex} if the intersection with any
axis-parallel line has at most one connected component.  The simplest  yet
unresolved setting involves orthogonal 6-gons, which include \emph{L-shapes}. For examples, consider \Cref{fig:Lshapes}.
\begin{figure}[htb]
  \centering
  \begin{subfigure}[t]{0.3\linewidth}
    \centering
    \includegraphics[page=1]{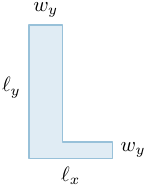}
    \subcaption{}
    \label{fig:LshapesA}
  \end{subfigure}
  \hfil
  \begin{subfigure}[t]{0.3\linewidth}
    \centering
    \includegraphics[page=2]{figures/Lshapes.pdf}
    \subcaption{}
    \label{fig:LshapesB}
  \end{subfigure}
  \hfil
  \begin{subfigure}[t]{0.3\linewidth}
    \centering
    \includegraphics[page=1]{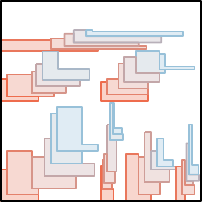}
    \subcaption{}
    \label{fig:LshapesC}
  \end{subfigure}
  \caption{(a) An L-shape, (b) a symmetric L-shape, and (c) a bin
  packing of L-shapes.}
  \label{fig:Lshapes}
\end{figure}

The diverse applications motivate a multitude of interesting packing problems.
Among them are bin packing, strip packing, perimeter packing and area packing
for various types of objects and allowed motions; for definitions see
\cref{sec:prelim}.  Many variants of these packing problems are known to be
computationally difficult.  While NP-hardness has usually been known for a long
time \cite{DBLP:journals/eatcs/Alt16,fowler1981optimal,harren,leung1990packing},
more recently, some variants were proved to even be $\exists\mathbb R$-complete
by Abrahamsen, Miltzow, and Seiferth~\cite{abrahamsen2024framework}.  Adding to
this difficulty, we might not have complete information about the objects in
advance, e.g., they may appear over time. This is captured by the so called
\emph{online setting} where the objects are presented one after another and the
next object is only revealed when the previous one has been placed.

In \emph{competitive analysis}, the performance of  an online algorithm is
measured against the optimal offline solution. While the  \emph{absolute
competitive ratio} gives the worst-case performance over all possible inputs,
the \emph{asymptotic competitive ratio} describes the behavior when the input
size goes to infinity. For precise definitions, we refer to
\autoref{sec:prelim}.  In the following we present our contributions, see also
\Cref{table:BPresults,tab:2}.  Usually, we present  lower bounds on the
asymptotic competitive ratio and upper bounds on the absolute competitive ratio
which imply the statements for both ratios.

\ExplSyntaxOn

\cs_new:Nn \algo_typeset_cell:n {
  #1
}
\cs_new:Nn \algo_shortref:n {
  \group_begin:
  \tl_set:Nn \theoremautorefname {Thm}
  \tl_set:Nn \corollaryautorefname {Cor.}
  \autoref{#1}
  \group_end:
}
\NewDocumentCommand{\TypesetCell}{m}{\algo_typeset_cell:n {#1}}
\NewDocumentCommand{\shortref}{sm}{\IfBooleanTF{#1}{[\algo_shortref:n
{#2}]}{\algo_shortref:n {#2}}}
\ExplSyntaxOff
\begin{table}[b]
  \centering
  \caption{Bounds on the best asymptotic competitive ratios of bin
  packing orthogonal polygons.}
  \label{table:BPresults}
  \begin{tblr}{
      colspec={|cl|cr|cr|},
      cell{1}{1} = {c=2}{c},cell{1}{3} = {c=2}{c},cell{1}{5} = {c=2}{c},
      rowspec={|Q|QQQQ|[dashed]QQ|},
    }
    Objects && Lower Bound && Upper Bound &\\
    \SymbRect{} & rectangles & 1.91 &\cite{epstein2019lower} & 2.56
    &\cite{han2011new} \\
    \SymbSymL{} & symmetric L-shapes & 1.91 &\cite{epstein2019lower}
    & 41 &[\shortref{thm:CCAlgSym}c] \\
    \SymbL{} & orthogonal  6-gons & $\nicefrac{n}{\log n}$
    &\shortref*{thm:LBbin} & $n$ & \\
    \SymbZ{} & orthogonal  8-gons & $n$ &\shortref*{thm:Zgons} & $n$\\
    \SymbLFrame{} & \Lframes & $1.54$& \cite{balogh2021new}
&  2 &\shortref*{thm:skeletons} \\
    \SymbZFrame{} & \Zframes & $n$ &\shortref*{thm:skeletons} & $n$  &
  \end{tblr}

\end{table}

\subparagraph{Bin Packing of 6-gons}
For bin packing of orthogonal polygons, it is natural to consider the
variants where the objects can only be translated,
and where the
objects may be rotated by multiples of 90 degrees.
The following proposition shows that even in the latter variant
we may restrict our analysis to the translational case. In the problem
\ortho, we  receive  orthogonal polygons $P_1, . . . ,
P_n$ which have to be placed into a unit square bin by translation. 
The goal is to minimize the number of used bins. \rot denotes the
variant  where the polygons might be rotated by multiples of 90 degrees.

\begin{restatable}{proposition}{rotation}\label{prop:equiv}
  The best (absolute/asymptotic) competitive ratios 
  for \ortho and \rot are within constant factors of each other. The
  analogous statements hold for the best approximation ratios of the
  offline variants.
\end{restatable}

\begin{proof}
  While each solution of \ortho is clearly a valid packing for \rot,
  we can transform a valid packing for  \rot into a valid packing for
  \ortho by splitting each bin into four such that each resulting bin
  contains polygons of the same rotation. Thus, we have
  $\opt_\text{rot}\leq \opt_\text{trans}\leq 4\cdot
  \opt_\text{rot}$. These insights directly imply that for every
  function $f$, an $f$-competitive/approximation algorithm for one
  problem yields an algorithm with competitive/approximation ratio
  $4f$ for the other problem.
\end{proof}

\cref{prop:equiv} allows us to focus on translations only.
 We start with a lower bound on the
competitive ratio of any online algorithm for \Lshapes.
To this end, let \trans denote the variant of \ortho where each
polygon in the sequence is an  \Lshape.

\begin{restatable}{theorem}{lowerBoundBin}\label{thm:LBbin}
  The asymptotic competitive ratio of any online
  algorithm for \trans is in $\Omega(\nicefrac{n}{\log n})$.
\end{restatable}

We obtain this lower bound by reducing from an online sorting problem
where natural numbers have to placed in arrays such that the numbers
in each array are increasing.
Note that the trivial algorithm, which packs each object into its own
bin, is $n$-competitive.

In contrast to \cref{thm:LBbin}, there exist constant competitive algorithms for
some interesting subclasses of \Lshapes, namely when the \Lshapes are either all
symmetric or all small. An \Lshape is described by four parameters $\ell_x,
\ell_y, w_x, w_y\in[0,1]$. It is \emph{small}  if $\ell_x,\ell_y\leq
\nicefrac{1}{2}$, \emph{large} if $ \ell_x, \ell_y\geq \nicefrac{1}{2}$, and
\emph{symmetric} if $\ell_x=\ell_y$ and $w_x=w_y$. For illustrations see
\cref{fig:Lshapes}.

\begin{restatable}{theorem}{onlineSym}\label{thm:CCAlgSym}
  There is an online algorithm for \trans with a constant asymptotic 
  competitive ratio when all \Lshapes are symmetric or small.
  In particular, there exist algorithms that use at most 
  \begin{alphaenumerate}
  \item \label{item:lasy}
    $33\cdot\OPT + 2$ bins
    when all \Lshapes are large and symmetric.
  \item \label{item:small}  $8\cdot \opt +7$
    when all \Lshapes are small.
  \item \label{item:sym} $41\cdot \opt +9$
    when all \Lshapes are symmetric.
  \end{alphaenumerate}
\end{restatable}

\cref{thm:CCAlgSym} directly  yields algorithms with a constant absolute
competitive ratio of 35, 15 and 50, respectively, as in bin packing, we have
	$\OPT \geq 1$ for any nonempty set of items.

The results for symmetric \Lshapes are based on identifying their connection to
scheduling problems.  Competitive algorithms for two special cases of large and
symmetric \Lshapes follow directly from previous work on scheduling problems.
Precisely, an $e$-competitive algorithm for objects of equal
widths~\cite{devanurOnlineAlgorithmsMachine2014} and a 2-competitive algorithm
for the case of equal lengths~\cite{epsteinOpenEnd} can be deduced, see
\autoref{sec:relWork} for explanations.  A connection to scheduling problems is
also used for the 33-competitive algorithm which  builds on a constant
competitive algorithm by Devanur et
al.~\cite{devanurOnlineAlgorithmsMachine2014} for a special machine minimization
problem in combination with an online coloring algorithm for intervals graphs by
Kierstead and Trotter~\cite{kierstead1981extremal}. For details on the
connection to scheduling problems, we refer to \autoref{sec:relWork}.

\subparagraph{Other packing variants}
 Our insights on online bin packing of
L-shapes have consequences for other packing variants.  In the problem  \peri,
we  receive  \Lshapes $L_1, . . .  , L_n$, and the goal is to place
them in the plane such that the perimeter of the bounding box is minimized. In
\minarea, the goal is to minimize the area of the bounding box.  For a unit
square bin and \Lshapes with arm lengths $\ell_x,\ell_y$ bounded by some
constant $t < 1$, the \emph{online critical density} is the largest value of $A$
such that there exists an online algorithm which packs any sequence  with  total
area at most $A$ into the bin.

\newcommand{\StatementConseqI}{The asymptotic competitive ratio of any
  online algorithm for \transStrip is in $\Omega(\nicefrac{n}{\log n})$.
}

\newcommand{\StatementConseqII}{
  	There is an online algorithm for \peri with absolute
  competitive ratio in $O(1)$.
}

\newcommand{\StatementConseqIII}{There is an online algorithm for \minarea with asymptotic
  competitive ratio in~$O(\sqrt{n})$.
}

\begin{restatable}{corollary}{consequences} \label{thm:conseq}
  The following statements hold for \Lshapes with arm lengths in $[0,1]$:
  \begin{romanenumerate}
  \item \label{item:i} \StatementConseqI
  \item \label{item:ii} \StatementConseqII
  \item \label{item:iii} \StatementConseqIII
  \item \label{item:iv} The online critical packing density is
    positive if the arm lengths of each \Lshape are bounded by $t<1$.
    For $t=1$, the online critical packing density is $0$.
  \end{romanenumerate}
\end{restatable}

The $O(\sqrt{n})$-competitive algorithm for \minarea  is
asymptotically best-possible, as a lower bound of $\Omega(\sqrt{n})$
already for rectangles has been shown by Abrahamsen and
Beretta~\cite{AbrahamsenBeretta20}.

\newcommand{\bigroot}[1][]{\Omega\bigg(\sqrt[#1]{\frac{\log n}{\log
\log n}}\bigg)}
\newcommand{\tblref}[1]{\text{#1}}
\begin{table}[htb]
  \centering
  \caption{Overview on the known upper and lower bounds on the best
    competitive ratios for the various packing problems, and on the
  critical densities.}
  \label{tab:2}
  \footnotesize
  \begin{tblr}{
      colspec={clrlrlrlr},
      cell{1}{2,4,6,8} = {c=2}{c},
      cell{4}{1} = {valign=m},
      cell{5}{6} = {c=2}{c},
      column{1-Z} = {leftsep=3.2pt,rightsep=3.2pt},
      column{2,4,6,8} = {rightsep=0pt},
      column{3,5,7,9} = {leftsep=0pt},
      cell{2-Z}{2-Z} = {mode=math},
      hline{1,2,Z} = {solid},
      vline{1,2,4,6,8,10} = {solid},
      width = \textwidth
    }
    Packing Variant &
    \SymbRect{} rectangles &&
    \SymbConv{} convex polygons &&
    \SymbL{} L-shapes &&
    \SymbZ{} ortho polygons \\

    {Bin Packing} &
    \Theta(1)&\tblref{\cite{jansen2009two}} &
    \bigroot&\tblref{\cite{onlineSortingOnlinePacking_SODA}}&
    \Omega\big(\frac{n}{\log n}\big)&\tblref{[Thm \ref{thm:LBbin}]} &
    n & \\

    {Strip Packing} &
    \Theta(1)&\tblref{\cite{YeOnlineStrip}} &
    \bigroot, O(n^{0.59})\,&\tblref{\cite{onlineSortingOnlinePacking_SODA}} &
    \Omega\big(\frac{n}{\log n}\big)&\tblref{[Thm \ref{thm:conseq}\ref{item:i}]} &
    \Omega\big(\frac{n}{\log n}\big)&\tblref{[Thm \ref{thm:conseq}\ref{item:i}]} \\

    {Perimeter Packing} &
    \Theta(1)&\tblref{\cite{AbrahamsenBeretta20}} &
    \Omega\bigg(\!\kern-1pt\mathord{\sqrt[4]{\frac{\log n}{\log \log n}}}\bigg)
     &\tblref{\cite{onlineSortingOnlinePacking_SODA}} &
    \Theta(1)&\tblref{[Thm \ref{thm:conseq}\ref{item:ii}]} \\

    {Area Packing} &
    \Theta(\sqrt{n}) & \tblref{\cite{AbrahamsenBeretta20}} &
    \Omega(\sqrt{n}) & \tblref{\cite{AbrahamsenBeretta20}} &
    \Theta(\sqrt{n})\,\tblref{\tblref{\cite[Thm \ref{thm:conseq}\ref{item:iii}]{AbrahamsenBeretta20}}}&&
    \Omega(\sqrt{n}) & \tblref{\cite{AbrahamsenBeretta20}} 
  \end{tblr}
\end{table}

\newcolumntype{s}{>{\centering\arraybackslash}X}

\subparagraph{Orthogonal polygons of higher complexity} For  bin
packing of orthogonal 8-gons, we show that the trivial algorithm,
that places each shape into its own bin,  achieves the best possible
competitive ratio. Hence, the trivial algorithm is also best possible
for \ortho.

\begin{restatable}{theorem}{Zgons}\label{thm:Zgons}
  There is no online algorithm for \ortho that achieves an
  asymptotic competitive ratio better than $n$,
  even when restricting to orthogonally convex 8-gons. 
\end{restatable}

\begin{figure}[htb]
  \centering
  \begin{subfigure}[t]{0.32\linewidth}
\includegraphics[page=2]{figures/zframes/zframes3.pdf}
    \hfil
    \subcaption{}
    \label{fig:ZshapesA}
  \end{subfigure}
  \hfil
  \begin{subfigure}[t]{0.32\linewidth}
\includegraphics[page=1]{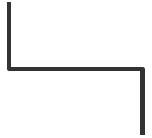}
    \hfil
    \subcaption{}
    \label{fig:ZshapesB}
  \end{subfigure}
  \hfil
  \begin{subfigure}[t]{0.32\linewidth}
\includegraphics[page=3]{figures/zframes/zframes3.pdf}
    \hfil
\subcaption{}
    \label{fig:ZshapesC}
  \end{subfigure}

  \caption{(a) A Z-shape (b) A \Zframe (c) An \Lframe}
  \label{fig:Zshapes}
\end{figure}

We obtain \cref{thm:Zgons} by considering Z-shapes of equal
thickness; \cref{fig:ZshapesA} depicts an example.
Interestingly, we obtain this result via first considering packings
of skeletons, which can be viewed as orthogonal polygons of widths 0,
see \cref{fig:ZshapesB,fig:ZshapesC}.
In the online problem \Lskeleton (\Zskeleton), a sequence of \Lframes
(\Zframes) shall be packed into a
minimum number of bins. In this context, we consider a packing of
skeletons as \emph{valid} if any two  objects are disjoint, with the
exception that either of the two endpoints of one skeleton may
touch another; a point of a \Zframe is an endpoint if it is a unique
extreme point, i.e., topmost, bottommost, leftmost, or rightmost.
From L- to \Zframes, there is a quite  drastic jump in complexity.

\newcommand{\StatementZskelLB}{No online algorithm for \Zskeleton has an
asymptotic competitive ratio better than $n$.}

\begin{restatable}{theorem}{skeletons}\label{thm:skeletons}
  \label{thm:frames}
  For online bin packing of orthogonal skeletons
the following holds.
  \begin{romanenumerate}
  \item There exists an online algorithm for \Lskeleton with a
  constant absolute competitive ratio.
  \item \StatementZskelLB
  \end{romanenumerate}
\end{restatable}

\subparagraph{Organization}
The remainder of the paper is organized as follows.  In
\cref{sec:relWork}, we discuss related work. In
\cref{sec:prelim}, we give a concise dictionary on crucial notions.
In \cref{sec:OnlineBinSort}, we prove \cref{thm:LBbin}, i.e.,  the
lower bound on the competitive ratio for \trans and
use our insights to show a lower bound for \transStrip in
\cref{sec:stripPacking}.
In \cref{sec:ConstantCompetitive}, we prove \cref{thm:CCAlgSym} by showing upper bounds on the competitive
ratios for \trans when the input shapes are symmetric or small.
We prove \cref{thm:conseq} in \cref{sec:variants}, giving lower and upper bounds
for the competitive ratios of other packing types.
We consider L-skeletons in \cref{sec:Lskeletons} and show Theorem~\ref{thm:skeletons}(i).
Finally, in \cref{sec:Zshapes}, we consider
orthogonal polygons of higher complexity and present the proofs of
Theorems~\ref{thm:Zgons} and~\ref{thm:skeletons}(ii). 

\subsection{Related work}\label{sec:relWork}

As the literature on online packing problems is extensive, we refer
to the surveys of Christensen, Khan, Pokutta, and
Tetali~\cite{CHRISTENSEN201763}, Epstein and van
Stee~\cite{epstein2018multidimensional}, van
Stee~\cite{DBLP:journals/sigact/Stee12,DBLP:journals/sigact/Stee15},
and Csirik and Woeginger~\cite{Csirik1998} for an overview.
Below we discuss the most important results related to the packing
problems studied in this paper.

\subparagraph{Rectangles}

Rectangle packings has been studied for almost 50 years. A
straightforward reduction from partition shows that the offline strip
packing problem cannot be approximated with an absolute factor
better than 3/2 unless $\P=\NP$, and the best known  approximation ratio is
$5/3+\varepsilon$ by Harren, Jansen, Prädel, and Van Stee~\cite{harren}.
For the online variant, first fit shelf algorithms as proposed by
Baker and Schwarz~\cite{baker1983shelf}, are widely studied. The best
known competitive ratio of $\nicefrac{7}{2}+\sqrt{10}\approx 6.66$ is
achieved by  Ye, Han, and Zhang~\cite{YeOnlineStrip}.
Restricting the attention to large instances, the best asymptotic
competitive ratio is
lower bounded by $1.54$ by van Vliet~\cite{van1992improved}
and
upper bounded by $1.59$ by Han, Iwama, Ye, and Zhang~\cite{han2007strip}.

For \emph{bin packing} of rectangles under translation, there exist approximation algorithms with approximation ratio 2~\cite{jansen2009two} and 
 asymptotic approximation guarantee~$1.4$ by Bansal and Khan~\cite{bansal2014improved}.
For online bin packing of rectangles, the upper bound on the
asymptotic competitive ratio for online translational bin packing
axis-parallel rectangular pieces into unit square bins has been
decreased in a series of papers from
$3.25$~\cite{coppersmithMultidimensionalOnLineBin1989a} to about
$2.55$ by Han, Chin, Ting, Zhang and Zhang~\cite{han2011new} and the
lower bound has been increased from $1.6$ \cite{galambos} to $1.91$
by Epstein~\cite{epstein2019lower}.
For the special case of squares, the lower bound is approximately
1.75 due to Balogh, B{\'e}k{\'e}si,  D{\'o}sa, Epstein, and
Levin~\cite{balogh2019lower}, while the upper bound is above
2.08 by Epstein and Mualem~\cite{DBLP:journals/algorithmica/EpsteinM23}.

For the problem of area minimization, Abrahamsen and
Beretta~\cite{AbrahamsenBeretta20} present $O(n)$-competitive
algorithms and show a lower bound of $\Omega(n)$
for both versions of rotation allowed and translations only.
For perimeter minimization, they gave a $3.98$-competitive algorithm for both, translation and rotation, as well as a lower bound on the
competitive ratio of $\nicefrac 43$ for the case of translation and
of $\nicefrac 54$ for the case of rotation.

\subparagraph{Convex Polygons}
Constant-factor approximations for packing convex polygons exist for many
variants~\cite{alt_convexOffline_JoCG,alt_convexOffline_corr,onlineSortingOnlinePacking_SODA,convexOffline-improved};
an exception is bin packing where  constant approximation factors  are only
known if the diameter of the items is bounded by a constant.  In contrast to
these constant guarantees, for online translational packings of convex polygons,
Aamand, Abrahamsen, Beretta and Kleist~\cite{onlineSortingOnlinePacking_SODA}
show that strip packing, bin packing, and perimeter packing do not allow for
constant competitive online algorithms, even if all pieces have a diameter
bounded by an arbitrarily small constant. 
To this end, they reduce to an online sorting problem which 
sparked follow-up
investigations~\cite{onlineSortingOnlineTSP,DBLP:journals/corr/abs-2508-14287,DBLP:conf/waoa/NirjhorW25}.
 Concerning upper bounds, they give
an online algorithm with competitive ratio $O(n^{0.59})$ for strip packing.
The knapsack problem for convex polygons under rigid motions is studied by
Merino and Wiese~\cite{DBLP:conf/icalp/MerinoW20}.

\subparagraph{Orthogonal Polygons and Scheduling}

We identified interesting connections between packings of special
\Lshapes and  scheduling problems such as open end bin packing and
machine minimization. These yield algorithms for large and symmetric
\Lshapes in the case where all \Lshapes have equal lengths or equal
widths; for an overview see \cref{tab:3}.

The case of equal lengths is related to open end bin packing.  In a variant of
classical 1D bin packing known as \emph{ordered open end bin packing},  items of
sizes in (0, 1] are presented one by one, and must be assigned to bins in this
order. An item can be assigned to any bin for which the current total size is
strictly below 1, i.e.,  the bin can be overloaded by its last packed item.
Balogh, Epstein, and Levin \cite{baloghMoreOrderedOpen2020} present a
$2$-competitive online algorithm and Epstein~\cite{epsteinOpenEnd} designs a
$1.5$ approximation. As packings of large, symmetric \Lshapes of equal lengths
can be reduced to open end bin packing, we obtain offline and online algorithms
with constant factors.  Specifically, for packing \Lshapes of length $\ell$, we
consider open end bins of capacity $1-\ell$. An \Lshape $L_i$ with arm length
$\ell$ and width $w_i$ is associated with a real $r_i:=w_i$. It is easy to see
that an $L_i$ can be packed into a bin containing $L_1,\dots, L_k$
(bottom-leftmost) if   $\sum_{j=1}^k w_j +\ell\leq 1$, \new{see also \cref{fig:equalLengths}}. Similarly, $r_i$ can be
packed into an open bin containing $r_1,\dots, r_k$ if (and only if)
$\sum_{j=1}^k r_j \leq 1-\ell$. As these conditions are equivalent, solutions of
open end bin packing translate to (bottom-leftmost) packings of large, symmetric
\Lshapes of equal lengths.  This implies a 2-competitive online algorithm and a 1.5-approximation \new{in this case.} 

\begin{figure}
    \centering
    \includegraphics[page=10]{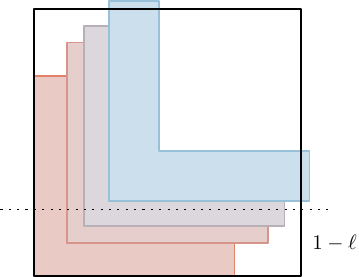}
    \caption{%
        Bin packing of large and symmetric \Lshapes of lengths $\ell$ reduces to ordered open end bin packing.
    }
    \label{fig:equalLengths}
\end{figure}

For bin packing of large and symmetric L-shapes \new{with different lengths}, we identified a close
connection to the problem of  \emph{machine minimization}: Given a set of jobs,
each of which has a {release time}~$r$, a {processing time} $p$ and {deadline}
$d$, the task is to schedule these jobs on a minimum number of machines such
that each machine processes at most one job at a time and each job starts not
before its release time, and finishes not after its deadline. The connection is
as follows: For large and symmetric \Lshapes, we show  in 
 \cref{sec:ConstantCompetitive} 
that the two parameters, length~$\ell$ and
width~$w$, translate to a job with release time $0$, processing time $w$, and
deadline $1-\ell$. \new{Therefore, we present relevant work in the context of machine minimization.}

Yu and Zhang~\cite{yuSchedulingMinimumNumber2009} present approximation
algorithms for the special cases of equal deadlines and equal processing times.
In particular, their 2-approximation algorithm for the special case of equal
deadlines translates to a 2-approximation for offline bin packing of large,
symmetric \Lshapes.  Kao, Chen,  Rutter, and Wagner \cite[Lemmas 1 and
2]{kao2012competitive} present a polynomial algorithm that computes the offline
optimal solution for jobs of unit processing times. This translates to a
polyomial time algorithm for offline packings of instances where all \Lshapes
are large, symmetric, and have equal widths. 

Devanur, Makarychev, Panigrahi, and Yaroslavtsev \cite[Section
2,3]{devanurOnlineAlgorithmsMachine2014} present an $e$-competitive online
algorithm for unit processing times; they also show that the competitive ratio
of $e$ is optimal. This result is also implied by Bansal, Kimbrel,  and
Pruhs~\cite[Lemmas 4.7 and 4.8]{bansal2007speed}. The insights can be translated
to an $e$-competitive algorithm for packing large, symmetric \Lshapes of equal
widths.  Moreover, Devanur et al.~\cite{devanurOnlineAlgorithmsMachine2014}
present a 16-competitive online algorithm for the case with equal deadlines;
their techniques are  the basis of our 33-competitive algorithm for packing large
symmetric \Lshapes, namely \cref{thm:CCAlgSym}(\ref{item:lasy}).

There are many more results on the general case of machine minimization.
However, due to a lower bound of $n$  by Saha \cite{sahaRentingCloud2013}, most
subsequent work  allows preemption of the jobs~\cite{chenLogCompetitiveAlgorithm2018, chenPowerMigrationOnline2016,
imLogLogCompetitive2018}, a setting that does not carry over to packings of
\Lshapes.

\begin{table}[htb]
  \caption{Upper bounds on the best asymptotic
  approximation/competitive ratios for \trans with  restricted \Lshapes.}
  \label{tab:3}
  \centering
  \begin{tblr}{
      colspec={|l|Q[c,$]r|Q[c,$]r|},cell{1}{2} = {c=2}{c},cell{1}{4} = {c=2}{c},
      rowspec={|Q|QQQQ|},
}
Type &
    \text{Offline} &&
    \text{Online}&
    \\
    \textbf{large, symmetric} &
    2 &\cite{yuSchedulingMinimumNumber2009} &
    33&\text{[Thm\ \ref{thm:CCAlgSym}\ref{item:lasy}]}
    \\
    \quad + equal length &
    1.5&\cite{epsteinOpenEnd} &
    2&\cite{baloghMoreOrderedOpen2020}
    \\
    \quad + equal width &
    1&\cite{kao2012competitive} &
    e& \cite{devanurOnlineAlgorithmsMachine2014}
    \\
    \textbf{small} &
    8&\text{[Thm\ \ref{thm:CCAlgSym}\ref{item:small}]} &
    8 &\text{[Thm\ \ref{thm:CCAlgSym}\ref{item:small}]}
  \end{tblr}
\end{table}

\subparagraph{Critical densities}
For a given container of volume $C$ and a class of objects in
$\mathbb R^d$, the \emph{critical density} is the largest value of
$V/C$ such that any sequence of objects of the class with a total
volume of at most $V$ can always be packed in the container.
For convex bodies of bounded diameter in $\mathbb R^3$, Auerbach,
Banach, Mazur and Ulam~\cite{mauldin2020scottish} stated (without
proof) that the critical density is positive, i.e., there  exists a
function $f$ such that any sequence of convex bodies in $\mathbb
R^3$, each of diameter $\leq \delta$ and total volume of at most $V$,
can be packed into a cube with side length $s=f (\delta,V)$ when
rotations are allowed.
The first proof, even for arbitrary dimension, is given by
Kosi{\'n}ski~\cite{kosinski1957proof}.
When rotations are not allowed, Alt, Cheong, Park, and
Scharf~\cite{alt2019packing} showed that the critical density of
packing convex bodies of bounded diameter into a cube is $0$; in
particular for packing unit disks in $\mathbb R^3$.
However, for cubes, Moon and Moser~\cite{moon1967some} showed that
any sequence of cubes in $\mathbb R^d$  of total volume $1/2^{d-1}$
can be packed into the unit cube. This bound is the best possible,
because for any $\varepsilon>0$, two cubes with side lengths
$\nicefrac{1}{2}+\varepsilon$ cannot be packed in the unit cube.

The study of critical densities likewise makes sense when the pieces
appear in an online fashion.
A lower bound on the critical density of online packing squares into
the unit square has been improved in a sequence of
papers~\cite{januszewski1997line,Brubach2014improved,han2008online,fekete2017online}
from $5/16$~\cite{januszewski1997line} to $2/5$~\cite{Brubach2014improved}.
Interestingly, Januszewski and Lassak~\cite{januszewski1997line}
proved that in dimension $d\geq 5$, the critical density of online
packing cubes into the unit cube is $1/2^{d-1}$, just as in the offline case.
Lassak and Zhang~\cite{lassak1991line} showed that for some constant
$\delta(d)>0$, any sequence of axis-parallel boxes of diameter and
total area at most $\delta(d)$ can be packed online in the
$d$-dimensional unit hypercube using translations. This implies that
the critical packing density for convex bodies is positive for any
dimension $d\geq 1$ when rotations are allowed as each body can be
rotated so that it has a constant (depending on the dimension)
density in its axis-parallel bounding box.
In contrast, for  translational and online packing convex polygons,
Aamand, Abrahamsen, Beretta and
Kleist~\cite{onlineSortingOnlinePacking_SODA} show that the critical
packing density is $0$.

\section{Notions: packing variants, competitive analysis, and
objects}\label{sec:prelim}
We briefly define the packing variants, review the common terminology for
competitive analysis, and introduce the considered objects.

\subparagraph{Packing variants}
Depending on the context, there are numerous
packing variants specified by the container and objective. The variants
discussed in this paper are as follows. Consider a set (or sequence) of items.
Given an unbounded supply of identical unit bins, the goal  of  \emph{bin
  packing} is to pack the items into as few bins as possible.  Given a
  horizontal strip of height $1$ which is bounded to the left but  infinite to
  the right, the aim of \emph{strip packing} is to place the items in the
  strip such that the maximum $x$-coordinate of an occupied point is minimized.  In
  \emph{perimeter packing} or \emph{area packing}, the goal is to find a
  placement of the items in the plane such that the perimeter or the area of
the bounding box is minimized, respectively.
All of these variants are interesting as
\emph{offline} and \emph{online} problems. In this work we focus on 
online problems and therefore review competetitive analysis.

\subparagraph{Competitive analysis}
In an online problem, the input is a sequence $\sigma_1,\ldots,\sigma_n$ of
objects, and we need to process object $\sigma_i$ (in a problem-specific
manner) before the next object $\sigma_{i+1}$ is revealed.  In this work, the
objects are either natural numbers, orthogonal polygons, or their skeletons.  We
briefly revisit the standard terminology of competitive analysis for an online
algorithm~$\alg$ of a minimization problem.  For an instance $I$, $\OPT(I)$
denotes the cost of the offline optimum solution and $\alg(I)$ denotes the cost
of the solution computed by $\alg$ for input~$I$.  Let $f\colon \mathbb N \to
\mathbb R$ be a function of the size $|I|$ of the instance, in our context
typically the number of objects.  We say that $\ALG$ has \emph{(absolute)
competitive ratio} $f(\ensuremath{|I|})$ if, for all instances $I$, it holds
that \( \ALG(I)\leq f(\ensuremath{|I|})\cdot \OPT(I).\) Similarly, $\ALG$ has
\emph{asymptotic competitive ratio} $f(\ensuremath{|I|})$ if there exists a
constant $\beta>0$ such that \( \ALG(I)\leq f(\ensuremath{|I|})\cdot
\OPT(I)+\beta\) for all $I$. 

\subparagraph{Objects} 
We focus on packings of
orthogonal polygons that we call \emph{L-shapes} and \emph{Z-shapes}. 
An \emph{\Lshape} is the union of a rectangle of width $\ell_x$ and height $w_y$
with a rectangle of width $w_x$ and height $\ell_y$, with $w_x \leq
\ell_x$ and $w_y \leq \ell_y$,
such that their lower left corners coincide. \cref{fig:shapesA}
presents an example. We describe the placement of an \Lshape by specifying
the coordinates of its \emph{reference point}, namely its lower left corner.  We
note that each orthogonal 6-gon is either an L-shape or a rotated L-shape.  An
\Lshape is 
\emph{small} if $\ell_x, \ell_y\leq \nicefrac{1}{2}$, \emph{large}
if $\ell_x, \ell_y\geq \nicefrac{1}{2}$, and \emph{symmetric} if $\ell_x=
\ell_y$ and $w_x=w_y$.

\IncludeSubfigures {shapes} {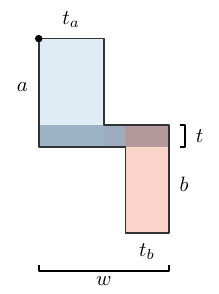} {2,3,1,4} {
Parameters and the reference point of an (a) \Lshape, (b) \Lframe, (c) Z-shape, and (d)~\Zframe.}

A \emph{Z-shape} is the union of three rectangles: A base rectangle of height
$t$ and width $w$, and two arm rectangles: One arm (above) of width $t_a$ and
height $a$, whose lower left corner coincides with the lower left corner of the
base, and one arm (below) of width $t_b$ and height~$b$ such that its upper
right corner coincides with the upper right corner of the base, where $t_a, t_b
\leq w$, see also \autoref{fig:shapesC}. In this paper, we will mostly be
concerned with Z-shapes of \emph{equal thickness}, i.e., the case where $t_a =
t_b = t$.  Note that all L-shapes and Z-shapes are \emph{orthogonally convex},
i.e., the intersection with any axis-parallel line is connected.

\Lframes and \Zframes are the corresponding shapes with zero thickness, i.e.,
$w_x = w_y = 0$ and $t_a = t_b = t = 0$, respectively.
\cref{fig:shapesB,fig:shapesD} present examples.

\section{Lower bound for bin packing of \Lshapes{} -- Proof of
\autoref{thm:LBbin}}
\label{sec:OnlineBinSort}

We now prove the lower bound for \trans.  \lowerBoundBin*

We show this in two steps. We prove that any online algorithm for \trans yields
an algorithm for an online sorting problem called \sorting with the same
competitive ratio, and show that every online algorithm for \sorting has a
competitive ratio in $\Omega(\nicefrac{n}{\log n})$. We start by introducing the
problem \sorting.

\subsection{Bin sorting}

The online problem \sorting{}$[k]$ can be described as a game between two
players, \playerA and \playerB, playing on several arrays with $k$ slots  each,
in which natural numbers have to be inserted. Within each array, the numbers
must be increasing. \playerA presents distinct numbers. When receiving a number,
\playerB has to  irrevocably decide in which array and in which slot the number
is placed while  obeying the rule that the numbers of each array are increasing.
The goal of \playerB is to minimize the number of used arrays, while \playerA
wants to maximize it.
\cref{{fig:instance_of_sorting_game}} presents an example of this game.

\begin{figure}[htb]
  \centering
  \includegraphics{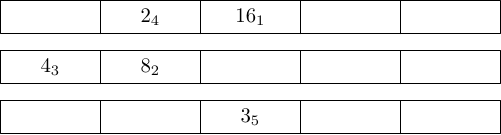}
  \caption{An instance of the \sorting with $k = 5$ for the sequence
  $16,8,4,2,3$.}
  \label{fig:instance_of_sorting_game}
\end{figure}

The offline optimum corresponds to the setting where \playerB knows all
numbers in advance and is simple to determine:  For $n$ numbers in arrays with
$k$ slots, $ \lceil \nicefrac{n}{k} \rceil$ arrays are necessary and, as the
numbers can be inserted in a sorted manner,
also  sufficient. Consequently, $\opt = \lceil \nicefrac{n}{k}
\rceil$.

We start by presenting a strategy for \playerA forcing \playerB to use at most
a logarithmic number of slots in each array.

\begin{lemma}
  \label{lem:LBsorting}
  The asymptotic competitive ratio of any
  online algorithm for
  \sorting{}$[k]$ is at least $\lceil
  \nicefrac{n}{\lfloor\log(k+1)\rfloor}\rceil /\lceil \nicefrac n
  k\rceil\in\Omega(\nicefrac{k}{\log k})$. For $k=n$, the lower bound is
  $n/{\lfloor \log (n+1) \rfloor} \in\Omega(\nicefrac{n}{\log n})$.
\end{lemma}
\begin{proof}
  We show that \playerA has  a strategy to present $n$ numbers such
  that \playerB
  uses at least $n/\lfloor\log (k+1)\rfloor$ arrays  while $\lceil \nicefrac n k
  \rceil$ arrays would  suffice.  The idea is as follows: 
  Consider a partially filled array $A$. The numbers in $A$
  	partition it into maximal sections
  	of free slots. If a new number shall be inserted in $A$,
   it may be placed into at most one of these sections, since
   the entries in the array must be increasing.
We say that two (possible future) numbers are \emph{similar} with respect to 
$A$ if
  they must be placed in the same section. The idea of \playerA is to maintain a
  set of active numbers that are similar for all arrays.

  In order to present $n$ natural numbers $a_1,\dots, a_n$, \playerA
  maintains the
  following invariant: In the beginning of iteration $i$, it has an
  \emph{active}
  set of $2^{n-i+1}-1$ similar numbers. In the very beginning (of iteration 1),
  the set is $(0,2^n):=\{1,\dots, 2^n-1\}$, each array has one
  section containing
  all slots, and \playerA presents $a_1=2^{n-1}$. In iteration~$i+1$, \playerA
  determines the next number to be presented based on the action of \playerB in
  iteration $i$. \playerB places $a_i$ in a section of some array
  which results in
  two smaller sections. If the smaller of the two section lies to the left of
  $a_i$, then the active set $(l,r)$ is updated to $(l,a_i)$ and
  $a_{i+1}=(l+a_i)/2=a_i-2^{n-(i+1)}$; otherwise it is updated to $(a_i,r)$ and
  $a_{i+1}=(a_i+r)/2=a_i+2^{n-(i+1)}$. As the new active set is a subset of the
  previous, \playerB must place numbers presented in the future in (subsets of)
  the current sections of the active set. By construction, the
  numbers of the active set remain similar.  The active set is non-empty for
  at least $n$ iterations because it contains $2^{n - i + 1} - 1\geq
  1$ numbers if
  $i\leq n$.

  When we restrict our attention to a fixed array with $k$ slots,
  then the section
  for the active numbers is (at least) halved each time a number is inserted;
  this is because we ensure that the active set must be inserted in the smaller
  section.  Therefore, \playerB  can place at most $\lfloor \log{(k+1)}
  \rfloor$ of the $n$ numbers in each array and thus needs at least
  $\lceil \nicefrac{n}{\lfloor\log (k+1)\rfloor}\rceil$ bins. This yields a
  competitive ratio of at least $\lceil \nicefrac{n}{\lfloor\log
  (k+1)\rfloor}\rceil/ \lceil \nicefrac n k \rceil\in\Omega(\nicefrac{k}{\log
  k})$.
\end{proof}

We remark that the lower bound of \cref{lem:LBsorting} is tight, that is,
there exists an online algorithm \sorting{}$[k]$ that matches this bound.
\begin{lemma}
  \label{lem:LBsortingUB}
  There exists an online algorithm for \sorting{}$[k]$ with
  absolute competitive ratio
  of $\lceil \nicefrac{n}{\lfloor\log(k+1)\rfloor}\rceil /\lceil \nicefrac n
  k\rceil\in O(\nicefrac{k}{\log k})$.
\end{lemma}
\begin{proof}
  For each partially filled array and any number, the valid slots form a
  consecutive interval.  Clearly, in an empty array, all slots are valid. Let
  $\mathcal A$ denote the algorithm that places a new number in the first
  array with a valid slot, and in particular, in the middle slot of the
  section. With
  this choice, the smallest section of an array containing $m$ numbers has
  length $\geq \lfloor\nicefrac{(k+1)}{2^m}\rfloor-1$. Therefore each array
  contains at least $\lfloor\log(k+1)\rfloor$ numbers before a new array is
  opened. Consequently, $\mathcal A$ uses at most $\lceil
  \nicefrac{n}{\lfloor\log(k+1)\rfloor}\rceil$ arrays while $\opt=\lceil
  n/k\rceil$.
\end{proof}

\color{black}

\subsection{From \trans to \sorting}
We now show how to  create instances of \trans which simulate \sorting.
The construction is visualized in \Cref{fig:to_sorting_game}. For
fixed $k \in \mathbb N$,
we consider the family $\mathcal L^k = \{ L_i \}_{i \in \mathbb N}$
of \Lshapes  with the following parameters:
\[\ell_x(L_i)=\nicefrac{1}{2}+\nicefrac{1}{2k}, \quad
  w_x(L_i)=\nicefrac{1}{2k},\quad \ell_y(L_i)=1-2^{-i-1}, \quad w_y(L_i)=2^{-i-1}\]

\begin{figure}[htb]
  \centering
  \begin{subfigure}[t]{0.3\linewidth}
    \centering
    \includegraphics[page=1,scale=.93,max
    width=\linewidth]{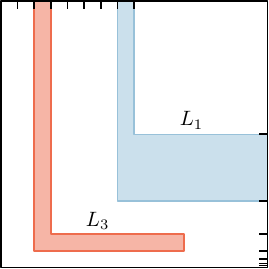}
    \subcaption{}
  \end{subfigure}
  \hfill
  \begin{subfigure}[t]{0.3\linewidth}
    \centering
    \includegraphics[page=2,scale=.93, max width=\linewidth]{figures/Lower_Bound_L-2.pdf}
    \subcaption{}
  \end{subfigure}
  \hfill
  \begin{subfigure}[t]{0.3\linewidth}
    \centering
    \includegraphics[page=3,scale=.93, max
    width=\linewidth]{figures/Lower_Bound_L-2.pdf}
    \subcaption{}
    \label{fig:LLBSortedOrder}
  \end{subfigure}
  \caption{Illustrations for $\mathcal L^8$. (a) While $L_1$ may be
    placed to the right/above of $L_3$, (b) $L_1$ cannot be placed to
    the left/below of $L_3$ as $l_x(L_1) + l_x(L_3) > 1$ and $l_y(L_3)
    + w_y(L_1) = 1 - 2^{-3 - 1} + 2^{-1 - 1} > 1$, respectively. (c) A
  packing of $L_1, \ldots, L_k \in \mathcal L^k$ for $k=8$ into one bin. }
  \label{fig:to_sorting_game}
\end{figure}

The key idea is that  $L_i$ in \trans mimics the natural number $i$
for \sorting. We start with the crucial property that each subset of
$ \mathcal L^{k}$ packed into a bin is sorted.

\begin{lemma}
  \label{th:l_shapes_sorted}
  Each set $S \subset \mathcal L^{k}$  packs into one bin if and only
  if $|S| \leq k$ and the items are placed in descending order of
  their indices $i$.
\end{lemma}

\begin{proof}
  By construction up to $k$ items of $\mathcal L^{k}$ fit into one
  bin in descending order of their indices:
  With respect to the $x$-coordinate, exactly $k$ items can be packed
  next to each other.
  For the $y$-coordinate, we note the following. When $L_{i}$ is
  placed somewhere in the bin, the free space above is lower bounded by
  $\ly(L_i) - \wy(L_i)=(1-2^{-i-1}) -2^{-i-1}= \ly(L_{i-1}) \geq
  \ly(L_j)$ for $j<i.$
  So when placed in descending order, there is space for $k$ items.
  \Cref{fig:LLBSortedOrder} depicts an example for $k = 8$.
  It remains to show that  $L_i$ cannot be placed below/left of $L_j$
  when $i < j$.
  It cannot be placed entirely to the left of $L_j$, because
  $l_x(L_i) + l_x(L_j) = 2(\nicefrac12 + \nicefrac1k) > 1$,
  and it cannot be placed below $L_j$ because
  $\wy(L_i)+\ly(L_j)=2^{-i-1}+1- 2^{-j-1}>1$.
\end{proof}

We now show how an algorithm for \trans yields an algorithm for \sorting{}$[k]$.
\begin{lemma}\label{lem:packingToSorting}
  An online-algorithm for \trans yields an algorithm for
  \sorting{}$[k]$ with the same (absolute/asymptotic) competitive ratio.
\end{lemma}
\begin{proof}
  Let $\mathcal A$ be an $\alpha$-competitive algorithm for \trans
  and consider an instance $I_S$ of \sorting$[k]$ with $n$ numbers.
  We describe an online algorithm $\mathcal A'$ for \sorting{}$[k]$
  which uses $\mathcal A$ as a subroutine. For each presented number
  $i$, $\mathcal A'$ presents $L_{i}\in\mathcal L^k$ to $\mathcal A$.
  By \Cref{th:l_shapes_sorted}, at most $k$ {\Lshapes} of $\mathcal
  L^k$ fit into one bin.
  Let  $x$ denote the {$x$-coordinate of the reference point} of
  $L_{i}$ in the $j$th bin as assigned by $\mathcal A$. Then
  $\mathcal A'$ places number~$i$ in the  slot $\lfloor x \cdot
  2k\rfloor+1$ of the $j$th array. As $x\in [0,\nicefrac 1
  2-\nicefrac{1}{2k}]$, it holds that $\lfloor x \cdot 2k\rfloor+1\in [k]$.
  We shortly argue that this position is valid. Firstly, {because all
  \Lshapes have width} $\nicefrac{1}{2k}$, the assigned $x$-values
  differ by at least $\nicefrac{1}{2k}$ and for any {$x$ and $x'$
  with} $|x-x'|\geq \nicefrac{1}{2k}$, we have $|\lfloor x \cdot
  2k\rfloor-\lfloor x' \cdot 2k\rfloor|\geq 1$. Secondly, as
  $\mathcal A$ guarantees that the numbers of each array are
  increasing, the same holds for the indices of the \Lshapes.  Hence,
  by \Cref{th:l_shapes_sorted}, the placement is valid.
  Consequently, we may translate the $x$-position of $L_{i}$ in the
  $j$th bin to a unique slot in the $j$th array.
The number of used arrays equals the number of used bins. Moreover,
  the offline optimum is $\lceil n/k \rceil$ in both problems. Hence,
  $\mathcal A'$ is $\alpha$-competitive.
\end{proof}

Together, \cref{lem:LBsorting,lem:packingToSorting} show the
following statement, which proves  \cref{thm:LBbin}.

\begin{corollary}\label{cor:bin}
  For each online algorithm $\mathcal A$ for \trans and each $n$,
  there exists a strategy to present $n$ \Lshapes from $\mathcal L^n$
  (no duplicates) such that they can all be packed into one bin while
  $\mathcal A$  uses at least $n/\log (n+1)$ bins.
\end{corollary}

\section{Strip packing of \Lshapes{} -- Proof of
\autoref{thm:conseq}(\ref{item:i})}\label{sec:stripPacking}
We now prove \Cref{thm:conseq}(\ref{item:i}). To this end, we show
how a strip packing algorithm yields a bin packing algorithm.

\begin{lemma}\label{lem:strip}
  An algorithm $\alg$ for strip packing a sequence of $n$ \Lshapes
    from~$\mathcal L^n$ (without duplicates) with asymptotic competitive
    ratio $\alpha$ yields an algorithm $\mathcal A'$ for \trans with
  asymptotic competitve ratio $2\alpha$.
\end{lemma}

\begin{proof}
  Let $\mathcal A$ be an asymptotic $\alpha$-competitive algorithm for
  \transStrip and consider an instance  of \trans with $n$ objects from
  $\mathcal L^n$. We describe an online algorithm~$\mathcal A'$ for \trans which
  uses $\mathcal A$ as a subroutine.  For each each $L\in \mathcal L^n$,
  $\mathcal A'$ presents it to $\mathcal A$ and observes its placement in the
  strip. The idea is to partition the strip into slots of widths
  $\nicefrac{1}{2n}$; each slot will correspond to a placement in a bin. For an
  illustration, see \autoref{fig:stacks}. We show that objects of $n$
  consecutive slots can be packed into one bin. If the reference point of $L$ is
  placed at $(x,y)$ in the strip, then $\mathcal A'$ places it as follows. Let
  $m,r\in \mathbb N$ such that $\lfloor x\cdot 2n\rfloor=mn+r$ and $r<n$. Then
  $\mathcal A'$ places $L$ in the $(m+1)$st bin such that its reference point is
  at $(x',y'):=(r/2n,1-\ell_y(L))$.

  \begin{figure}[htb]
    \centering
    \includegraphics[page=1,scale=.93]{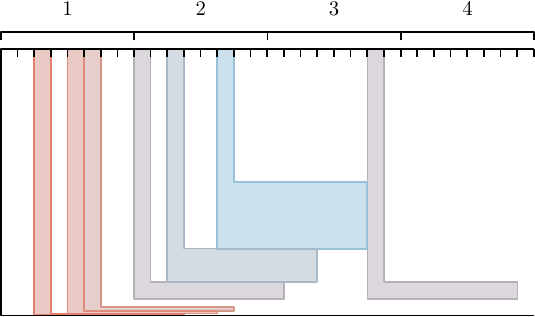}
    \caption{Illustration for the proof of \cref{lem:strip}.
        {
            The position assignments made by an online algorithm for \transStrip
            can be used to construct an online algorithm for \trans.
        }
    }
    \label{fig:stacks}
  \end{figure}

  We now argue that this yields a valid bin packing.  Let $L_i$ and $L_j$ be two
  \Lshapes  in a same bin and denote their $x$-coordinates in the strip by $x_i$
  and $x_j$, and in the bin by $x'_i$ and $x'_j$, respectively. We assume that
  $x_i<x_j$. By construction, we have that $|x_i-x_j|\leq \nicefrac{1}{2}$. The
  placement $x_i<x_j$ implies that $w_y(L_i)+\ell_y(L_j)\leq 1$ and hence $i\geq
  j$, i.e., in each bin the \Lshapes are non-increasing. Moreover, due to the
  equal width of $\nicefrac{1}{2n}$, no two \Lshapes intersect in the
  $x$-direction.  Without duplicates, no intersections in the $y$-direction
  follow from the definition of $\mathcal L^n$. Lastly, if $x'_j-x'_i\leq
  \nicefrac{n-1}{2n}$ then $x'_j+\ell_x(L_i)\leq
  (\nicefrac{n-1}{2n}+x_i)+\nicefrac{1}{2}+\nicefrac{1}{2k}\leq x_i+1$.
  Therefore, the $y$-span of the \Lshapes is at most 1.

  It remains to check the number of bins. Let $\OPT_S$ and $\OPT_B$ denote the
  offline optimum values for \transStrip and \trans, respectively. Clearly, we
  have $\OPT_S\leq \OPT_B=1$; recall that all
  objects pack into one bin by \cref{th:l_shapes_sorted}.  As $\mathcal A$ is asymptotically
  $\alpha$-competitive, there is a constant $\beta > 0$ such that the width of
  the strip packing is at most $\alpha \cdot \OPT_S + \beta$. The number of bins
  is $\lceil{\lceil 2n (\alpha \cdot \OPT_S + \beta) \rceil}/{n}\rceil\leq
  2\alpha\OPT_B + 2\beta + 2$, so $\mathcal A'$ is asymptotically
  $2\alpha$-competitive.
\end{proof}

Together, \cref{lem:strip,cor:bin} imply \cref{thm:conseq}(\ref{item:i}).

\section{Constant competitive algorithms for symmetric and small
\Lshapes{} -- Proof of \autoref{thm:CCAlgSym}}
\label{sec:ConstantCompetitive}

In this section we present constant competitive algorithms, i.e., we
prove  \cref{thm:CCAlgSym}.
\onlineSym*

We show the three parts in three lemmas. We start by  establishing
some useful properties of packings of large \Lshapes.

\subsection{Useful properties of special packings}
We call a bin packing of \Lshape{s} \emph{stacked}, when the arms of
each \Lshape can be extended to the boundary of the bin without
intersecting any other \Lshape. For illustrations see \cref{fig:usefulA,fig:usefulB}.
Note that in a stacked packing, the $x$- and $y$-order of the
\Lshapes coincide.
We call a packing of \Lshape{s} a \emph{gravity packing}, if no
\Lshape can be moved to the left or bottom without causing
intersections. While the packing in \cref{fig:usefulB} is not a
gravity packing, the one in \cref{fig:usefulC} is.

\IncludeSubfigures{useful}{figures/Lower_Bound_L-2}{7,5,6}[0.32,0.32,0.32]{
  Illustrations for stacked and gravity packings. (a) A packing that
  is not stacked. (b) A stacked packing that is not a gravity
  packing. (c) A stacked gravity packing.
}

\begin{lemma}\label{lem:stackedSymmetric}
  Let $S$ be a set of large \Lshapes. The following properties hold:
  \begin{romanenumerate}
  \item Any bin packing of $S$ is stacked.
  \item If all \Lshapes in $S$ are symmetric, then in any gravity
    packing of $S$, the reference point of each \Lshape is placed on
    the diagonal.
  \end{romanenumerate}
\end{lemma}
\begin{proof}
  We prove both statements individually.
  \begin{romanenumerate}
  \item Suppose there exists a valid packing where an (extended)
    $x$-arm of an \Lshape $L_i$ intersects an (possibly
    extended) $L_j$, see \cref{fig:usefulA}. Then $\lx(L_i) +
    \lx(L_j) \leq 1$; otherwise the \Lshapes do not fit in one
    bin. However, as both \Lshapes are large, we have
    $\lx(L_i), \lx(L_j) > \nicefrac{1}{2}$, a contradiction.
    Analogous arguments apply for the $y$-direction.
  \item Clearly, the reference point of the bottom-leftmost \Lshape
    lies on the origin and thus on the diagonal; otherwise some
    \Lshape can be shifted to the left or bottom. By symmetry, the
    reflex corner of the bottom-most (symmetric) \Lshape lies on the
    diagonal in a  gravity packing and all other \Lshapes lie above
    and to the right of this point. Hence, iteratively using this
    argument for the remaining square subcontainer, the claim follows. \qedhere
  \end{romanenumerate}
\end{proof}

We note the following useful property of stacked gravity packings.

\begin{lemma}\label{lem:adding}
  Given a stacked gravity packing of \Lshapes $L_1,\dots, L_k$ in a
  unit bin, an \Lshape $L$ can be placed on top of the stack if and
  only if $\sum_{i=1}^k\wx(L_i) + \lx(L) \leq 1$ and
  $\sum_{i=1}^k\wy(L_i) + \ly(L) \leq 1$.
\end{lemma}

\begin{figure}[htb]
	\centering
	\includegraphics{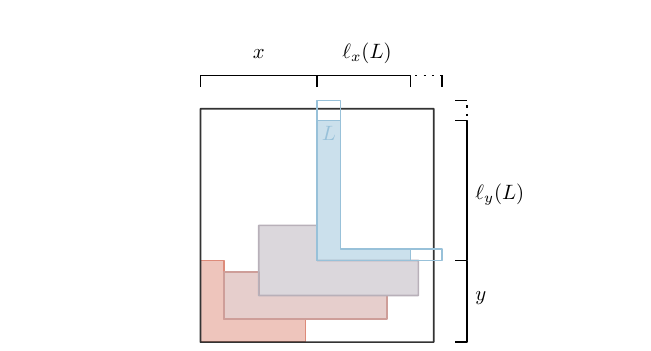}
	\caption{Illustration for \autoref{lem:adding}.  The \Lshape $L$
		can be added to the stack, if its arms and the stack are small
		enough to respect the bin capacity.  Conversely, if $L$ does not
		fit,  the bin size is exceeded in at least one dimension.}
	\label{fig:binfit}
\end{figure}

\begin{proof}
Clearly, the best placement of $L$ on the stack is leftmost and bottommost.
    In other words, we need to check whether $L$ can be placed at
    $(x, y)$ with $x \coloneqq \sum_{i = 1}^k \wx(L_i)$ and
    $y\coloneqq\sum_{i = 1}^k \wy(L_i)$, see \cref{fig:binfit} for an
    illustration.
    This is the case if and only if there is enough space to
    accommodate the arm lengths, i.e., if and only if $x + \lx(L)\leq
  1$ and $y + \ly(L)\leq 1$.
\end{proof}

\subsection{Large and symmetric \Lshapes}

We now show a constant competitive online algorithm for large and symmetric \Lshapes.

\begin{proposition*}[corresponding to \cref{thm:CCAlgSym}(\ref{item:lasy}) -- based on Theorem 1.2 of~\cite{devanurOnlineAlgorithmsMachine2014}]
    \label{lem:OnlineSymLarge}
	There is a 33-competitive online algorithm for \trans when all \Lshapes are large and symmetric.
\end{proposition*}
\newcommand{\longL}{long}
\newcommand{\shortL}{short}

Let $L_1,\dots,L_n$ denote a sequence of large and symmetric \Lshapes where each
$L_i$ is specified by an arm length $\ell_i$ and a width $w_i$.  We show how the
bin packing problem of large symmetric \Lshapes reduces to a one-dimensional
problem which has been studied in the context of machine minimization.  By
\cref{lem:stackedSymmetric}, each packing is stacked and the $x$- and
$y$-coordinate of each reference point are equal, see also
\autoref{fig:OnlineSymLargeIntervalsA}.

Hence, it suffices to specify the $x$-coordinate. Moreover, for each $L_i$, the
$x$-coordinates occupied by the vertical arm form an interval $[x_i, x_i + w_i]
\subseteq [0, 1-\ell_i+w_i]$, which we call the \emph{$x$-range} of $L_i$.  As
every packing is stacked, the $x$-ranges are interior disjoint, consider
\autoref{fig:OnlineSymLargeIntervalsA} for an illustration.  For each L-shape
$L_i$,  the rightmost valid $x$-coordinate of the reference point is $\hat x_i
:= 1 - \ell_i \in [0, 1]$.  In the corresponding one-dimensional problem of
packing the $x$-ranges, we seek a position $x_i \in [0, \hat x_i]$ for each $i
\in [n]$ and a bin such that no two (open) $x$-ranges of a bin intersect.

This problem can be rephrased as a job scheduling problem with deadlines, known
as \emph{online machine minimization}: Given a sequence of jobs, where each job
$j$ is specified by a processing time $p_j$, a release time $r_j$ and a deadline
$d_j$, the goal is to schedule the jobs on a minimum number of identical
machines. In a valid schedule, each job is assigned to a machine and a start
time $s_j$ such that $r_j\leq s_j$, $s_j+p_j\leq d_j$ and the (open) intervals
\IncludeSubfigures {OnlineSymLargeIntervals} {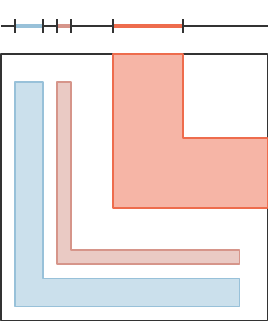} {1,2} [0.34,0.34] {
	Bin Packing of large symmetric L-shapes and the corresponding
	one-dimensional problem of packing $x$-ranges with deadlines.  (a) The
	$x$-ranges in for a bin packing of large symmetric L-shapes and (b) the
	corresponding $(\hat x + w)$-values, assigning each L-shape to an interval $I_k, k
	\in \mathbb N$.
}
$(s_j,s_j+p_j)$ assigned to the same machine are pairwise disjoint.

With our considerations above, the translation from our setting is
straight-forward: Each bin corresponds to a machine and each L-shape models a
job with processing time $w_i$, release time $0$, and deadline time $\hat x_i +
w_i$.

Devanur, Makarychev and Panigrahi~\cite{devanurOnlineAlgorithmsMachine2014}
present a constant-competitive online algorithm for the case where all deadlines
are equal and the jobs are presented in order of their release times. By
reversing the direction of time, essentially exchanging release times and
deadlines, the result holds equivalently for the case when all release times are
equal and the jobs are presented in the order of their deadlines. Hence, in our
setting, we  meet the first requirement.  However, the L-shapes may arrive in
any order, so the second requirement is not met and we cannot apply their
algorithm directly.  In following, we show how combining it with an online
coloring algorithm for intervals graphs by  Kierstead and Trotter~\cite[Theorem
5]{kierstead1981extremal} yields a constant-competitive algorithm for our
problem, where L-shapes may appear in any order.

The first step is to group the \Lshapes based on the rightmost valid point of
the $x$-ranges, i.e., $\hat x_i+w_i$. To this end, we partition $[0, 1]$ into
intervals. For $k \in \mathbb N$, we define $I_k = [\frac{1}{2^k}, \frac{1}{2^{k
- 1}})$; note that $|I_k| = 2^{-k}$.  An L-shape $L_i$ belongs to category $C_k$
if $\hat x_i + w_i \in I_k$, see also \autoref{fig:OnlineSymLargeIntervalsB}.
An L-shape $L_i$ in category $C_k$ is called \emph{\shortL} (in $C_k$) if $w_i <
|I_k|/4 = 2^{-k-2}$, and \emph{\longL} otherwise.

We now describe the online algorithm \LasylPacker for \trans when all \Lshapes
are large and symmetric: It classifies the \Lshapes according to the categories
as described, and packs  \shortL{} L-shapes and \longL{} L-shapes into different
bins.  The $x$-range of a \shortL{} \Lshape of category $C_k$ is placed into an
interval $I_{k + 1}$ of a bin for \shortL{} \Lshapes, using \textsc{FirstFit} as
a subroutine, i.e., it is placed leftmost into the first $I_{k + 1}$ it fits.
For a  \longL{} \Lshape $L_i$, we set $x_i := \hat x_i$. To determine the bin,
we maintain a conflict graph $G$.  The vertices of $G$ represent the (open)
$x$-ranges $(\hat x_i, \hat x_i + w_i)$ for each long \Lshape $L_i$; two
vertices share an edge, whenever two corresponding $x$-ranges intersect. Note
that, by definition, $G$ is an interval graph.  Kierstead and
Trotter~\cite[Theorem 5]{kierstead1981extremal}  present a 3-competitive
algorithm for online coloring of interval graphs.  We use this algorithm, called
\textsc{Colorer} in the following, to maintain a coloring of the interval graph.
As every color class is an independent set in the conflict graph, all L-shapes
that are assigned the same color can be placed into a common bin.  Let $\OPT_L$
denote the minimum number of bins required to pack the \longL{} \Lshapes, and
similarly $\OPT_S$ for the \shortL{} \Lshapes. We show that both subroutines are
constant competitive algorithms.

\begin{lemma}
    \label{lemma:OnlineSymLargeShort}
    For \shortL{} L-shapes, the algorithm \LasylPacker uses at most $6\OPT_S + 2$ bins.
\end{lemma}
\begin{proof} 
    Let $k\in\mathbb N$ and consider \Lshapes of category $C_k$. Let $\OPT_k$
    denote the minimum number of bins required to place the \shortL{} L-shapes
    of category $C_k$.  By definition of $C_k$, the $x$-range of any $L\in C_k$
    is contained in the interval $I_k^{+} \coloneqq [0, 2^{-k + 1})$ in any
    valid, and, in particular, any optimal placement.  For each  $L\in C_k$, the
    largest valid $x$-coordinate is in interval $I_k$, so any placement of the
    $x$-range in $I_{k + 1}$ is valid, see also  \cref{fig:OnlineSymLargeShort}.

    \begin{figure}[htb]
        \centering
        \includegraphics[page=2]{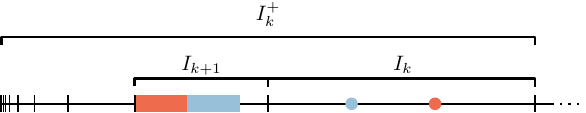}
        \caption{Illustration for \autoref{lemma:OnlineSymLargeShort}.
            {
                For two \Lshapes of category $C_k$, the rightmost valid point
                of their $x$-ranges is shown as a point.  The algorithm
                \LasylPacker places the corresponding $x$-ranges into
                interval $I_{k + 1}$.
            }
        }
        \label{fig:OnlineSymLargeShort}
    \end{figure}

    The algorithm only uses the interval $I_{k + 1}$ of each bin, and is
    therefore equivalent to bin packing where the bins are copies of $I_{k +
    1}$.  \textsc{FirstFit} fills all but at most two copies of $I_{k + 1}$ to a
    level of at least $\nicefrac23$: As long as a bin is filled less than
    $\nicefrac23$, items packed into any following bin have size at least
    $\nicefrac13$ of the bin size. As each item has size at most $\nicefrac12$
    of the bin size, at least two items fit into each bin. Hence, a new bin is
    opened only if the previous bin contains at least two items and is thus
    filled to at least $\nicefrac23$.  Because $|I^+_k |=2^{-k + 1} = 4|I_{k +
    1}|$, \textsc{FirstFit} fills all but at most two copies of $I_k^{+}$ to a
    level of at least $\nicefrac23\cdot \nicefrac14 = \nicefrac16$, i.e., it
    uses at most $6\OPT_k + 2$ bins for \shortL{} \Lshapes in category $C_k$.
        
    As the \Lshapes of different categories are placed into disjoint intervals
    and can thus be combined, \LasylPacker needs at most $\max_{k \in \mathbb N}
    6\OPT_k + 2 \leq 6 \OPT + 2$ bins in total.
\end{proof}

To analyze the \longL{} shapes, we consider, for $k \in \mathbb N$, the
intervals $I_k^{++} \coloneqq [0, \frac{1}{2^{k - 1}} + \frac{1}{4\cdot 2^k})$.
We show that in every valid placement every shape occupies a large portion of
$I_k^{++}$, or doesn't intersect $I_k$ at all.

\begin{lemma} 
    Let $k\in \mathbb N$. For each long $L_i$, whose $x$-range intersects $I_k$,
    the interval  $I_k^{++} \cap [x_i, x_i + w_i]$ has a size of at least
    $\nicefrac19|I_k^{++}|$.\label{lemma:OnlineSymLargeIkplus}
\end{lemma}
\begin{proof}
    By definition, $|I_k^{++}| = 2|I_k| + \nicefrac{1}{4}|I_k| =
    \nicefrac94|I_k|$.  Let $L_i\in C_j$, then $x_i+w_i\leq \hat x_i +
    w_i<2^{-j+1}$. We distinguish three cases.  If $j>k$, then the $x$-range of
    $L_i$ does not intersect $I_k$, as the rightmost valid position for the
    right endpoint of the $x$-range is in $I_j$, which is still to the left of
    $I_k$. For an illustration consider \autoref{fig:OnlineSymLargeIkplus}(a).  

    \begin{figure}[bh]
        \centering
        \includegraphics[page=1]{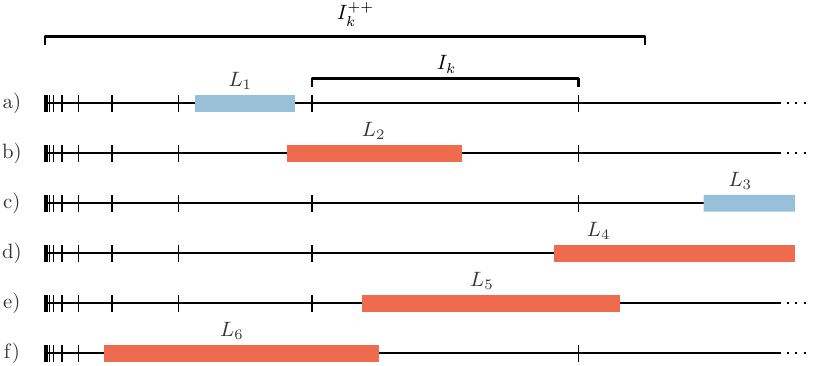}
        \caption{
            Illustration for \autoref{lemma:OnlineSymLargeIkplus}.  (a) The
            $x$-range of an \Lshape in $C_j, j > i$, doesn't intersect~$I_k$.
            (b) The $x$-range of an \Lshape in $C_k$ is contained in $I_k^{++}$.
            (c-f) In any other case, the $x$-range may end further to the right:
            (c) It may not intersect $I_k^{++}$ in which case it also doesn't
            intersect $I_k$. (d) The  $x$-range which may contain $I_k^{++} \cap
            I_{k - 1}$ or otherwise (e)/(f) the $x$-range may be contained in
            $I_k^{++}$.
        }
        \label{fig:OnlineSymLargeIkplus}
    \end{figure}
    
    If $j=k$, then $w_i\geq \frac14|I_k|$ because $L_i$ is long and the
    $x$-range is contained in $I_k^{++}$, as shown in
    \autoref{fig:OnlineSymLargeIkplus}(b).  Thus, we have $|I_k^{++} \cap [x_i,
    x_i + w_i]| = |[x_i, x_i + w_i]| = w_i \geq \nicefrac14|I_k| =
    \nicefrac19|I_k^{++}|$.

    If $j<k$, then $w_i \geq \nicefrac14|I_{k - 1}| = \nicefrac12|I_k|$.  If the
    $x$-range of $L_i$ intersects $I_k$, either the $x$-range of $L_i$ contains
    $I_k^{++} \cap I_{k - 1}$, for which $|I_k^{++} \cap I_{k - 1}| =
    \nicefrac14|I_k| = \nicefrac19|I_k^{++}|$, or its rightmost $x$-coordinate,
    $x_i+w_i$ is in $I_k^{++}$, so it is  contained in $I_k^{++}$, and
    $|I_k^{++} \cap [x_i, x_i + w_i]| = w_i = \nicefrac12|I_k| =
    \nicefrac29|I_k^{++}|$, as shown in
    \autoref{fig:OnlineSymLargeIkplus}(d)-(f).
\end{proof}

This gives us a lower bound for the optimal packing, from which we get a bound
on the number of required bins for \longL{} L-shapes.

\begin{lemma}
    For \longL{} L-shapes, the algorithm \LasylPacker uses at most
    $27\cdot\OPT_L$ bins.
    \label{lemma:OnlineSymLargeLong}
\end{lemma}
\begin{proof}
    The number of bins used by \LasylPacker corresponds to the number of colors
    \textsc{Colorer}  uses  to color $G$. As \textsc{Colorer}  is
    $3$-competitive \cite{kierstead1981extremal}, the number of bins is at most
    $3\cdot \chi(G)$.  Because $G$ as an interval graph is perfect, we have
    $\chi(G)=\omega(G)$.  Assuming $\omega(G)\leq 9\cdot \OPT_L$, the number of
    used bins is 
    \[
        3\cdot \chi(G)
        = 3\cdot \omega(G)
        \leq 3\cdot 9\cdot \OPT_L
        \leq 27\cdot  \OPT_L.
    \]

    It remains to show $\omega(G)\leq 9\cdot \OPT_L$.  Note that a set of
    pairwise intersecting intervals has a common point. It therefore suffices to
    bound the intervals present at each point. We show a slightly stronger
    statement, namely that for each $k$, the number of $x$-ranges intersecting
    $I_k$ is at most $9\cdot\OPT_L$.  To this end, let $J_k$ denote the set of
    L-shapes whose $x$-range $(\hat x_i, \hat x_i+w_i)$ intersects $I_k$.  By
    \autoref{lemma:OnlineSymLargeIkplus}, the $x$-range of any \Lshape in $J_k$
    occupies at least $\nicefrac19$ of the interval $I_k^{++}$.  Clearly, this
    fact remains true for any packing where the \Lshapes are possibly placed
    further to the left.  As all \Lshapes are placed rightmost by \LasylPacker,
    it follows that any packing of the \Lshapes in $J_k$ needs at least
    $|J_k|/9$ bins, i.e., $\OPT_L\geq |J_k|/9$.  Consequently,  $\omega(G)\leq
    \max _k |J_k| \leq 9 \cdot \OPT_L$. 
\end{proof}

Together \cref{lemma:OnlineSymLargeShort,lemma:OnlineSymLargeLong} yield
\cref{thm:CCAlgSym}(\ref{item:lasy}): Every shape is either \shortL{} or
\longL{}.  Hence, in total, \LasylPacker uses at most  \[(6\cdot \OPT_S + 2) +
27\cdot\OPT_L \leq 33\cdot\OPT + 2\] bins.  This finishes the proof.

\subsection{Small \Lshapes}
We now turn our attention to small \Lshapes and present a constant competitive
algorithm.

\begin{proposition*}[corresponding to
  \cref{thm:CCAlgSym}(\ref{item:small})]\label{lem:OnlineSymSmall}
  There is an 8-competitive online algorithm for \trans when
  all \Lshapes are small.
\end{proposition*}

The idea of the algorithm is as follows.  Firstly, it categorizes the \Lshapes
into classes, depending on its two arm lengths. Then, for each class, we show
how to pack the items into rectangles of constant density. In a last step, we
show that we can pack the resulting rectangles almost optimally.

We start by showing how to pack \Lshapes in rectangles to guarantee some
density.  To this end, we employ first fit algorithms. For the classical bin
packing problem in 1D, which we denote by \oneDBin, the online algorithm \FF
keeps a list of open bins, which is initially empty.  When a new item arrives,
\FF identifies the first bin into which the item can be packed and places it
leftmost into the bin. If no bin can accommodate the item, a new bin is opened
and the item is placed leftmost inside it. In the following, we use \FF or
slight variations as subroutines.

\begin{lemma}
    \label{lem:small_lshapes_density_gen}
    Let $A_x,A_y$ and $c,C$ be reals with $0\leq c\leq C\leq 1$, and
    $L_1,L_2,\dots$ be a sequence of \Lshapes where for $z\in\{x,y\}$ we have  $c
    \cdot A_z< \ell_z(L_i) \leq C\cdot A_z$.  Then, there exists an online
    algorithm that packs the sequence of \Lshapes in rectangles of size
    $(A_x\times A_y)$ such that the \Lshapes in every but the last rectangle
    have a total area exceeding $c(1-C)\cdot A_xA_y$.
\end{lemma}

\begin{proof}
    The online algorithm \FFL keeps a list of rectangles each containing a
    stacked gravity packing of \Lshapes. In particular, for each new \Lshape
    $L$, it identifies the first rectangle $R$ such that the reference point of
    $L$ can be packed in the reflex corner of the current topmost \Lshape of
    $R$.  Let $L'$ denote the \Lshape for which the last rectangle is opened.
    Then $L'$ does not fit into any previous rectangle $R$.  By
    \cref{lem:adding}, there exists $z\in\{x,y\}$ with $\ell_z(L')+\sum_{L \in
    R} w_z(L) > A_z$, see also \cref{fig:small}. In particular, this implies
    that $\sum_{L \in R} w_z(L) >  A_z-\ell_z(L')\geq (1-C)A_z$.  Let $\overline
    z$ be such that $\{z,\overline z\}=\{x,y\}$.  Then, for the total area of
    all \Lshapes in bin $R$ we have:
    \begin{align*}
        \sum_{L\in R}\area(L)
        &\geq \sum_{L \in R} l_{\overline z}(L) w_z(L)
        > c\cdot A_{\overline z} \sum_{L\in R} w_z(L)
        > c \cdot  (1-C) \cdot A_{\overline z}A_z
    \end{align*}
    This proves the claim.
\end{proof}
\IncludeSubfigures {small} {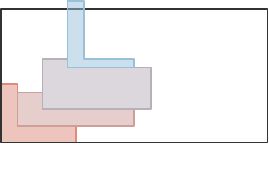} {1,2} [0.34,0.34] {
  Illustration for the proof of \Cref{lem:small_lshapes_density_gen}.
    When (a)  a new \Lshape does not fit into the rectangule, (b) the  already present \Lshapes occupy  a
    constant fraction of the rectangle's area.  
}

In order to make use of \cref{lem:small_lshapes_density_gen}, we partition the
\Lshapes as follows.  An \Lshape~$L$ belongs to category $C(a,b)$ if $2^{-a-2} <
\lx(L) \leq 2^{-a-1}$ and $2^{-b-2} < \ly(L) \leq 2^{-b-1}$. Because all
\Lshapes are small, we have $a,b\geq 0$.  With $c=\nicefrac{1}{4}$,
$C=\nicefrac{1}{2}$, $A_x=2^{-a}$, and $A_y=2^{-b}$,
\cref{lem:small_lshapes_density_gen} yields the following.
\begin{lemma}
    \label{lem:small_lshapes_density}
    The  algorithm \FFL packs  a sequence of \Lshapes from 
    $C(a,b)$  into
    $(2^{-a}\times 2^{-b})$-rectangles such that all but one rectangle have a
    packing density of~$\nicefrac{1}{8}$.
\end{lemma}

Next, we pack the resulting rectangles efficiently.  Let $\mathcal R$
denote the set of all $(a\times b)$-rectangles where $a,b\in
\{2^{-i}\mid i\in \mathbb N\}$.
We make use of the fact that each sidelength has the form $2^{-i}$
for some $i$. In particular, in 1D, \textsc{FirstFit} processes input
of this form optimally.  More generally, we call an instance of
\oneDBin \emph{strongly divisible}, if any item of size $s_i$ is a
divisor of any larger item of size $s_j$ and the bins capacity, i.e.,
there exist  $k_1,k_2\in \mathbb N$ such that $s_ik_1=s_j$ and $s_ik_2=1$.
Coffman, Gary, and Johnson showed that \FF is optimal for
divisible instances.
\begin{lemma}[\cite{coffmanBinPackingDivisible1987}, Theorem 3]
  \label{lem:divisible_sizes}
  For  strongly divisible instances, \textsc{FirstFit}  computes
  optimal solutions for \oneDBin. In particular, if a new bin for an
  item with size~$s$ is opened, then the total free space of all bins
  is smaller than $s$.
\end{lemma}

We employ this fact for packing the rectangles in $\mathcal{R}$.

\begin{lemma}\label{lem:packing_rectangles_optimal}
  For packing rectangles from $\mathcal{R}$ into unit bins,
  there exists an online algorithm which uses at most $\opt +
  3$ bins. In particular, the total free space in all bins  has an area
  of most $3$.
\end{lemma}

\begin{proof}
  The algorithm \textsc{NiceRectanglePacker} uses a two step
  approach. Firstly, it categorizes the rectangles of
  $\mathcal{R}$ by their height and packs rectangles of the same
  height $h$ into rectangular containers, which we call strips,
  of height $h$ and width $1$. Secondly, it packs the strips
  into unit bins. For an illustration, see \cref{fig:niceRectanglePacker}.

  \begin{figure}[htb]
    \centering
    \includegraphics[page=3,width=\textwidth]{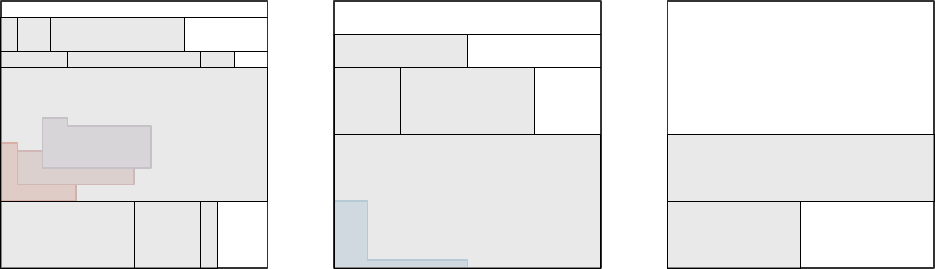}
    \caption{Illustration for {\textsc{NiceRectanglePacker} used in} the proofs of
      \Cref{lem:packing_rectangles_optimal} and
    \cref{thm:CCAlgSym}(\ref{item:small}).}
    \label{fig:niceRectanglePacker}
  \end{figure}

  Note that both problems reduce to \oneDBin, because on the one
    hand, the widths of the strips and bins are equal, and on the other
  hand, the height of the rectangles and strips are equal.
  For both tasks, we use \FF as a subroutine. In other
    words, a rectangle is placed in the leftmost position of one of the
    already open strips of the same height. When a new strip is needed,
    it is placed bottommost in one of the already open bins, if
  possible. A new bin is opened for the new strip otherwise. In each
  dimension, the size of the packed objects is  $2^i$ for some
  $i$ and
    \Cref{lem:divisible_sizes} is applicable to analyze both parts of
  the algorithm.

  We first consider the problem of
    packing strips of width $1$ into bins.
    We argue that the total free area in all bins is less than $1$.
    When a new bin is opened, then by \Cref{lem:divisible_sizes}, the
    total free area  in all old bins is $<h\cdot 1$. Moreover, the
    free area in the new bin is $(1-h)\cdot 1$. Consequently, the total
  free area is at most $1$.

  We now turn to the problem of
  packing rectangles of height $h$ into strips of height $h$. As
  for the bins, we may use \Cref{lem:divisible_sizes} to conclude that
  together all strips of height $h$ have free area of at most $1
  \cdot h$. Summing over all heights of the form $2^i$ for some
  $i\geq 0$, the free area in all strips is  upper bounded by
  $\sum_{i=0}^\infty \nicefrac{1}{2 ^i} =2$.

  So in total, the free space has area at most $1+2=3$.
  Clearly, the total area of the rectangles is a lower bound on
  $\OPT$. Hence, the claim follows.
\end{proof}

As our last ingredient, we need the following fact.
\begin{lemma}\label{lem:rectangleAreaTotal}
  \label{th:area_of_R}
  The total area of all rectangles in $\mathcal R$ is $4$.
\end{lemma}

\begin{proof}
  \cref{fig:perimeterA} depicts a tiling of the $2\times2$-square
  with the rectangles of $\mathcal R$ for $k=0$. In particular, the
  lower left point of an $(a\times b)$-rectangle is placed at
  $(a,b)$. It is easy to check that no two rectangles intersect
  interiorly. In particular, if the $x$-intervals  ($y$-intervals) of
  two rectangles intersect interiorly, then their start points must
  coincide. Because two rectangles intersect only if their $x$- and
  $y$-intervals intersect, this directly implies that no two distinct
  rectangles intersect.
\end{proof}

We are now ready to prove \cref{thm:CCAlgSym}(\ref{item:small}).

\begin{proof}[Proof of \cref{thm:CCAlgSym}(\ref{item:small})]
  The algorithm \smallLPacker classifies the \Lshapes into the
  categories $C(a,b)$ for $a,b\geq 0$.
Throughout the algorithm, we maintain an active rectangle for each
  category $C(a,b)$ of size $2^{-a}\times 2^{-b}$. The \smallLPacker
  proceeds as follows: If the appearing \Lshape can be packed
    into its active rectangle,  \FFL packs it there (as in
    \cref{lem:small_lshapes_density}). Otherwise, we close this
    rectangle, open a new rectangle,  use  \textsc{NiceRectanglePacker}
    to pack it into some bin, and declare it as active.
    \cref{fig:niceRectanglePacker} illustrates an example of a
  resulting packing where the \Lshapes of one category are depicted.

  Let $k$ denote the number of used bins and let $R,
    R_\text{cl},R_\text{ac}$ denote the set of all, closed, and active
  rectangles, respectively.
  By \cref{lem:packing_rectangles_optimal}, the total area of the
  rectangles is at least $k-3$.
  By \cref{lem:small_lshapes_density}, the density of each closed
  rectangle is $\nicefrac 18$.   The total area of active rectangles is $4$ by
  \cref{lem:rectangleAreaTotal}.
  Therefore, we have
  \[k-3\leq  \area(R)
    =\area(R_\text{cl})+\area(R_\text{ac})
    \leq 8\cdot \area(L)+4
  \leq 8 \cdot \opt +4.\]
  Consequently, \smallLPacker uses at most $8\cdot \opt +7$ bins.
\end{proof}

\subsection{Symmetric \Lshapes}
Together, \cref{thm:CCAlgSym}(\ref{item:lasy}) and (\ref{item:small})
imply  \cref{thm:CCAlgSym}(\ref{item:sym}).

\begin{proposition*} [corresponding to \cref{thm:CCAlgSym}(\ref{item:sym})]
  There is an online algorithm for \trans for symmetric L-shapes with
  asymptotic competitive ratio of 41.
  \label{thm:OnlineSym}
\end{proposition*}
\begin{proof}
  We partition the sequence into two subsequences, one containing only
  small and one containing only large \Lshapes. By
  \cref{thm:CCAlgSym}(\ref{item:lasy}) and (\ref{item:small}), the
  number of used bins is at most $(8\cdot \opt + 7)+(33\cdot \opt +
  2)=41\cdot \opt + 9$.
\end{proof}

\section{Other packing variants -- Proof of
\autoref{thm:conseq}}\label{sec:variants}
Our insights for bin packing also have consequences for other packing
variants. In this section, we prove the following corollary.
\consequences*

We have already presented a proof of
  \Cref{thm:conseq}(\ref{item:i}) in \cref{sec:stripPacking}. It
therefore remains to consider the remaining three statements.

\subsection{Packings in the plane minimizing objectives of the bounding box}

In this section, we prove \cref{thm:conseq}(\ref{item:ii}) and (\ref{item:iii}).
The idea is to use a two-step approach. Firstly, we employ
\Cref{lem:small_lshapes_density} to pack \Lshapes of similar size
into rectangles so that the packing density is a constant (in all but
one \emph{active} rectangle for each category). Secondly, we pack the
rectangles in the plane; here we make use of previous work by
Abrahamsen and Beretta~\cite{AbrahamsenBeretta20} .
We start to show that there exists a $O(1)$-competitive algorithm for
perimeter packing.

\begin{proposition*} [\cref{thm:conseq}\ref{item:ii}]
	\StatementConseqII
\end{proposition*}

\begin{proof}
  The idea is to employ \Cref{lem:small_lshapes_density} to pack
  \Lshapes of similar sizes into rectangles so that the packing
  density is a constant, namely $\nicefrac{1}{8}$ in all but one
  \emph{active} rectangle for each size class.
  In the second step, we pack the rectangles with the $4$-competitive
  algorithm \textsc{BrickTranslation}  by Abrahamsen and
  Beretta~\cite{AbrahamsenBeretta20}.

  To take care of the active rectangles, we use the following strategy.
  Let $\mathcal R$ denote the set of all $(a\times b)$-rectangles where
  $a,b\in \{2^{k}\mid  k\in \mathbb Z\}$. Note that the rectangles of
  $\mathcal R$ can be packed in the plane such that the $(a\times
  b)$-rectangle has its lower left corner at $(a,b)$, see also
  \cref{fig:perimeterA}. We claim that for any subset $R$ of $\mathcal
  R$, the perimeter of the bounding box is at most twice as large as an
  in an optimal packing. For an illustration, consider
  \cref{fig:perimeterB}. Let $A:=\max_{R\in\mathcal R} a(R)$ and
  $B:=\max_{R\in\mathcal R} b(R)$. Then the bounding box has perimeter
  at most $4(A+B)$. Moreover, the width is lower bounded by $A$ and the
  height by $B$. Hence the perimeter of any bounding box of $R$ is
  $2(A+B)$. Similarly, an \Lshape packed into an $(a\times
  b)$-rectangle has width at least $a/2$ and height at least $b/2$.
  Therefore, when viewing $R$ as a set of active rectangles which
  contains some \Lshapes $L$, any bounding box containing $L$ has
  perimeter $2(A/2+B/2)=A+B$. Therefore, we can pack the active
  rectangles with this strategy and obtain a $4$-competitive algorithm
  on the perimeter of the bounding box.

  \IncludeSubfigures {perimeter} {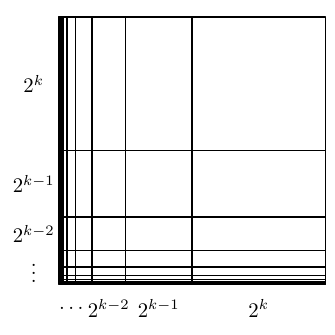} {1,2,3}
  [0.3, 0.3, 0.3] {
    Illustration for the proofs of \cref{lem:rectangleAreaTotal} and
    \cref{thm:conseq}\ref{item:ii}. 
      (a) A packing of the rectangles in $\mathcal R$. (b) The
      perimeter of the bounding box is at most $4(A+B)$. (c) Combining
    the packing of $\mathcal R$ with the brick packing.
    \label{fig:tiling}
  }

  We use the above insight as a tool for book keeping. The idea is to
  pack the first active rectangle (with constant packing density of
  less than \nicefrac{1}{8}) of each category with the above
  strategy. For all following rectangles of the same category, we use
  the algorithm \textsc{BrickTranslation} from Abrahamsen and
  Beretta~\cite[Section 2.1]{AbrahamsenBeretta20}. For the analysis,
  we pretend that all rectangles packed by \textsc{BrickTranslation}
  are inactive, i.e., in our minds, we exchange the active rectangle
  of each category with the inactive counter part in the $\mathcal
  R$-packing (if it exits).  Hence, all rectangles packed by
  \textsc{BrickTranslation} behave as rectangle, in particular the
  contained \Lshapes have a constant density and the perimeter of
  each item is a constant fraction of the perimeter of the rectangle.
  This allows to use the same analysis. We combine these two packings
  as illustrated in \cref{fig:perimeterC}, namely we mirror the
  brickpacking to ensure that small objects are close together. Let
  $W_1,H_1$ and $W_2,H_2$ denote the dimensions of the bounding boxes
  of these two parts. Clearly the combined bounding box has height
  $H:=H_1+H_2$ and width $W:=\max\{W_1,W_2\}$.
  We have $2(H_1+W_1)\leq 4\cdot \OPT_1$ and $2(H_2+W_2)\leq 8\sqrt
  2\cdot \OPT_2$.
  This yields $2(H+W)=2(H_1+H_2 +W) \leq 8\sqrt 2 +4 \cdot \OPT
  <15.32\cdot \OPT$.
\end{proof}

Now we show that there exists a $O(\sqrt{n})$-competitive algorithm
for min area packing.

\begin{proposition*} [\cref{thm:conseq}\ref{item:iii}]
  \StatementConseqIII
\end{proposition*}
\begin{proof}[Proof-Sketch]
  The idea is to employ \Cref{lem:small_lshapes_density} to pack
  \Lshapes of similar size into rectangles so that the packing
  density is a constant, namely $\nicefrac{1}{8}$ in all but one
  \emph{active} rectangle for each size class. Then we use the
  \textsc{DynBoxTrans} algorithm from Abrahamsen and
  Beretta~\cite[Section 3.2]{AbrahamsenBeretta20} to pack the
  rectangles in the plane. As in the case for rectangles, the total
  area and the product of the maximum width and maximum height is
  lower bound on the offline optimum. Together with the constant
  density in almost all rectangles, the analysis of Abrahamsen and
  Beretta carries over with tiny adjustments (in the definition of
    $T_j$ and for an analogous statement of Lemma 15 in
    \cite{AbrahamsenBeretta20}.
The modifications take care of the fact that, besides \emph{sparse}
    and \emph{dense} shelves, there are also \emph{active} shelves
    which contain an open rectangle that may not guarantee a constant
  density. For a simple analysis, we treat them as sparse shelves.)
  With this two step approach, we obtain a $O(\sqrt{n})$-competitive
  algorithm, which is asymptotically best possible already for rectangles.
\end{proof}

\subsection{The critical packing density is positive}

We now prove \cref{thm:conseq}(\ref{item:iv}).

\begin{proposition*}
  [\cref{thm:conseq}\ref{item:iv}]
  The online critical packing density is positive if the arm lengths of
  each \Lshape are bounded by $t<1$. For $t=1$, the online
  critical packing
  density is $0$.
\end{proposition*}
\begin{proof}
  First consider the case of $t=1$. Two \Lshapes $L$ and $L'$ with
  $\ell_x(L)=1$ and $\ell_y(L')=1$ cannot be packed into one bin. For
  any $A>0$,  the other parameters can be chosen such that the total
  area does not exceed $A$. Therefore, the critical packing density is $0$.

  For each $t<1$ we present  $A>0$ such that there exists an online
  algorithm that packs every sequence of \Lshapes with arm lengths
  bounded by $t$ and total area $A$ into the unit square.
  We set $A:=\nicefrac{1}{125} (1-t)^3$, $a:=\nicefrac{1}{10}(1-t)$
  and $h=2a$. We partition the \Lshapes into four categories $C_{ss},
  C_{sl}, C_{ls}, C_{ll}$, depending on whether the arm lengths are
  smaller or larger than $a$. For example, an $L$ is in $C_{sl}$ if
  $\ell_x(L)\leq a$ and  $\ell_y(L)> a$. For each category, we
  allocate a subcontainer in the unit bin as depicted in \Cref{fig:densityA}.

  \begin{figure}[htb]
    \centering
    \begin{subfigure}[t]{0.4\linewidth}
      \centering
      \includegraphics[page=1]{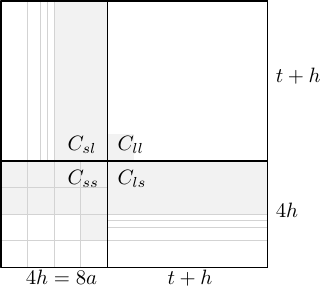}
      \subcaption{ $\nicefrac{1}{6}\leq t<1$}
      \label{fig:densityA}
    \end{subfigure}
    \hfil
    \begin{subfigure}[t]{0.4\linewidth}
      \centering
      \includegraphics[page=2]{figures/density}
      \subcaption{$t\leq \nicefrac{1}{6}$}
      \label{fig:densityB}
    \end{subfigure}
    \caption{Illustration for the proof of \cref{thm:conseq}\ref{item:ii}. 
        {
            In (a), a unit bin is subdivided such that every category can be packed into a separate region. If a better
            upper bound on the arm lengths of the \Lshapes is known, then a simpler packing strategy as in (b) can be
            used.
        }
    }
    \label{fig:density}
  \end{figure}

  For the \Lshapes in $C_{ss}$, we allocate an $8a\times 8a$-square
  in which we pack 16 $2a\times 2a$-squares. In particular, after
  scaling by $\nicefrac{1}{2a}$, we obtain a sequence of \Lshapes with
  arm lengths bounded by $ \nicefrac{1}{2}$ and total area of at most
  $\nicefrac A{(4a^2)}$. We use the algorithm of
  \cref{thm:CCAlgSym}(\ref{item:small}) to pack them into $k$ unit
  bins. As in the proof of \cref{thm:CCAlgSym}(\ref{item:small}), we
  have $k  \leq 8\cdot \nicefrac A{(4a^2)} + 7 \leq 8\leq 16$.

  For the \Lshapes in $C_{ll}$, we allocate a
  $(t+h)\times(t+h)$-square. We stack all \Lshapes of $C_{ll}$ in the
  sequence in a gravity packing, i.e., bottom and left most. Suppose an
  \Lshape $L'$ cannot be packed in the container. Then there exists
  $z\in \{x,y\}$ such that $\sum_ {L\in C_{ll} \setminus L'}w_z(L) +
  \ell_z(L')>t+h\implies \sum_ {L\in C_{ll} \setminus L'}w_z(L) =t+h-
  \ell_z(L')>h$.  As all arms have length at least $a$, the total
  area of $C_{ll}$ exceeds $\sum_ {L\in C_{ll}}w_z(L)\ell_{\overline
  z} (L)=ha=2a^2>A$, where $\overline z\in=\{x,y\}\setminus\{z\}$.

  For the \Lshapes in $C_{ls}$, we allocate a $(t+h)\times
  4h$-rectangle and partition them into height classes. An $L$ is in
  height class $H_i$, $i\in \mathbb N$, if $\nicefrac h 2 ^{i+1} \leq
  \ell_y\leq  \nicefrac h 2 ^{i}$. \Lshapes in $H_i$ are packed into
  $(t+h)\times\nicefrac h 2 ^{i}$-rectangles using \textsc{FirstFit}.  Note that
  $
a \leq \ell_x\leq t
$.

  For $t\geq \nicefrac{1}{6}$, we have $c:=\frac{a}{t+h}\leq
  \nicefrac{1}{4}$ and $C:=\frac{t}{t+h}\geq \frac 12$.
Therefore, \Cref{lem:small_lshapes_density_gen} guarantees a
  packing density of $c(1-C)$ in all but one \emph{active} rectangle
  of each height. The sum of heights of all active rectangles is
  $2h$. Therefore, \Lshapes of total area $2h(t+h)c(1-C)$ are packed
  by this subroutine. It holds that
  $2h(t+h)c(1-C)
\geq  \nicefrac{1}{125} (1-t)^3=A$.

  For $t\leq \nicefrac{1}{6}$, we have  $c=\nicefrac{1}{4}$, and
  $C=\nicefrac{1}{2}$ and thus obtain $2h(t+h)c(1-C)
\geq \frac{(1-t)}{5}
  \geq  \nicefrac{1}{125} (1-t)^3=A$.
  The \Lshapes in $C_{sl}$ are handled symmetrically.

  It follows that for each of the four categories, we can pack
  sequences of  \Lshapes with an area of at most $A$ and thus in
  particular any sequence where the total area is bounded by $A$. This
  finishes the proof.
\end{proof}

\begin{remark}
  We remark that the above bounds can easily be improved, e.g., for
  $t\leq \nicefrac 1 6$.  Suppose $t\leq \nicefrac{1}{2r}$ for some
  $r\geq 3$,  see \cref{{fig:densityB}} for an illustration of
  $r=3$. Then a sequence of \Lshapes with total area
  $\nicefrac{1}{8}(1 - \nicefrac{7}{r^2})\geq \nicefrac{1}{36}$ can
  always be packed into one unit bin, in particular into $r^2$ bins
  of size $\nicefrac1r\times\nicefrac1r$. For simplicity, we scale by
  a factor of $r$. Then  the arm lengths are bounded by $
  \nicefrac{1}{2}$ and the total area is at most $\nicefrac18({r^2-7})$.
  Let $k$ denote the number of bins that the algorithm of
  \cref{thm:CCAlgSym}(\ref{item:small}) needs. As in the proof of
  \cref{thm:CCAlgSym}(\ref{item:small}), we have $k  \leq
  8\cdot\nicefrac18({r^2-7}) + 7 =r^2$.
\end{remark}

We mention some upper bounds.
\begin{lemma}\label{lem:densityUB}
    For \Lshapes whose arm lengths {are} bounded by $t$ with $\nicefrac
  12<  t<1$, the  critical packing
  density is upper bounded by $2t(1 - t)$. For sequences of \Lshapes
  whose arm lengths {are} bounded by
  $t$ with $0\leq  t\leq\nicefrac 12$, the  critical packing density
  is upper bounded by
  $\nicefrac{2-2t}{2-t}$.
\end{lemma}
\begin{proof}

  Consider large  symmetric \Lshapes of arm length $t$ and width $w$.
  By \cref{lem:adding},  we can pack at most
  $\nicefrac{1-t}{w}+1$ shapes into a unit bin, and each shape has an
  area of $(2t-w)w$. Thus the
  total area is at most $(2t-w)(1-t+w)$, which gives an upper bound
  for the critical packing density
  for all $w > 0$, and converges to $2t(1-t)$ as $w$ tends to $0$.
  This is illustrated in \autoref{fig:densityUBA}.

  \begin{figure}[htb]
    \centering
    \begin{subfigure}[t]{0.35\linewidth}
      \centering
      \includegraphics[page=8]{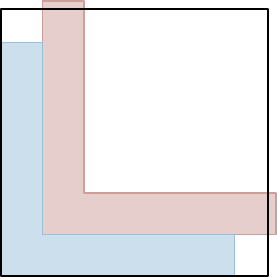}
      \subcaption{}
      \label{fig:densityUBA}
    \end{subfigure}
    \hfill
    \begin{subfigure}[t]{0.5\linewidth}
      \centering
      \includegraphics[page=9]{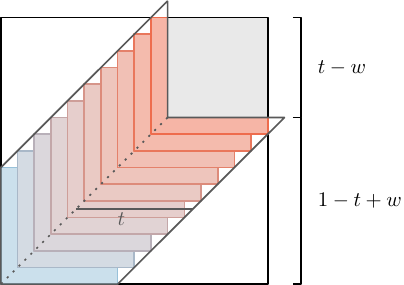}
      \subcaption{}
      \label{fig:densityUBB}
    \end{subfigure}
    \caption{Illustration for the proof of \cref{lem:densityUB}.
        {
           (a) When $t > \nicefrac12$, all \Lshapes in a common bin form a stack, so at most $\nicefrac{1-t}{w} + 1$ \Lshapes
            fit into each bin.
            (b) The area of a stack is at most $2(1-t-w)t$ and of the associated 
            empty region $(t - w)^2$.
        }
    }
    \label{fig:densityUB}
  \end{figure}

  For $t$ with $0<  t<\nicefrac 12$, consider  symmetric \Lshapes of
  arm length $t$ and width $w$.
  If a reference point of one \Lshape lies in the bounding box of
  another \Lshape, we say that
  these two \Lshapes are neighbors. A set of neighbors form a stack.
  For each stack, there is
  an empty region of area $(t-w)^2$ above the topmost \Lshape. The
  area of a stack of \Lshapes is
  bounded by $2(1-t+w)t$ which converges to $2(1-t)t$ as $w$ tends to
  $0$, see also \autoref{fig:densityUBB}. Consequently, the
  packing density is bounded by
  \[\frac{2(1-t)t}{2(1-t)t+t^2}=\frac{2(1-t)}{2-t}.\]
\end{proof}

\section{A constant competitive algorithm for bin packing of \Lframes{} -- Proof of \autoref{thm:skeletons}(i)}\label{sec:Lskeletons}

We now prove \cref{thm:skeletons}(i) by showing the following statement.
\begin{proposition*}[corresponding to \cref{thm:skeletons}(i)]
    There exists an  online algorithm for \Lskeleton  with absolute (and therefore also asymptotic) competitive ratio $2$. 
\end{proposition*}

\begin{proof}
    The proof is based on three simple insights. Firstly, we note that any
    finite set $S$ of \Lframes can be packed into one bin if for each $L\in S$
    we have $\ell_x(L),\ell_y(L)< 1$. The idea is to place the reference point
    of each item on the diagonal such that its longer arm touches the boundary,
    see \autoref{fig:LframeA}. If the point is already occupied, we place the
    \Lframe in the middle between this occupied point and the closest previous
    occupied point (or the origin).  
    
    Secondly, note that any two \Lframes $L,L'$
    with $\ell_x(L)=\ell_y(L')=1$ intersect, as illustrated in
    \autoref{fig:LframeB}.  
    
    Thirdly, for a set $S$ of \Lframes with  $\ell_x(L)=
    1$ for all $L\in S$, the problem reduces to \oneDBin, see also
    \autoref{fig:LframeC}.

    \IncludeSubfigures {Lframe} {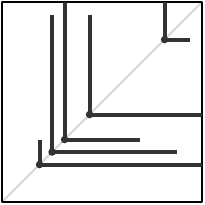} {1,2,3} {Illustration for the proof of \cref{thm:skeletons}(i).
    	\new{(a)  Any set of \Lframes with arm lengths $<1$ can be packed into one bin. (b) Any two \Lframes $L,L'$
    	with $\ell_x(L)=\ell_y(L')=1$ intersect. (c) A packing of \Lframes with $\ell_x(L)=1$ for all $L$  corresponds to \oneDBin.}
    }

    Therefore, we apply the following strategy. We partition a given set $S$ of
    \Lframes into three sets: An $L\in S$ is contained in  $S_1$ if
    $\ell_x(L),\ell_y(L)< 1$, in $S_2$ if $\ell_x(L)=1$ and in $S_3$ if
    $\ell_y(L)=1$. (An $L$ with $\ell_x(L)=\ell_y(L)= 1$ can be arbitrarily
    assigned to $S_2$ or $S_3$.)

    If $S_1\neq\emptyset$, we use one bin to pack the objects of $S_1$.  For $S_2$ and $S_3$, we apply
    \textsc{NextFit} for \oneDBin.  \textsc{NextFit}, studied by Johnson in
    \cite{johnson1973near}, has one active bin at the time. New arriving objects are typically placed in this active bin.  If a new item
    arrives that does not fit in the active bin, the current active
    bin is closed and and a active bin is opened. \textsc{NextFit} uses at most $2\opt(S_i)-1$ bins
    for $S_i$ with $i\in\{2,3\}$: Excluding the last bin, the sum of any two consecutive bins
    exceeds 1. Let $b$ denote the number of used bins by \textsc{NextFit}.
    Hence, $\opt(S_i)\geq b/2+1$ if $b$ is even and  $ \opt(S_i)\geq
    (b-1)/2+\varepsilon$ for any $\varepsilon>0$ if $b$ is odd. In both cases,
    we have $\opt(S_i)\geq (b+1)/2+1$ and thus $b\leq 2 \opt(S_i)-1$.

    In total, the algorithm uses at most $1+2(\opt(S_2)-1)+2(\opt(S_3)-1)\leq
    2(\opt(S_2)+\opt(S_3))\leq 2\opt$ bins.
\end{proof}

\section{Strong lower bound for orthogonal polygons of higher
complexity -- Proofs of \autoref{thm:Zgons} and \autoref{thm:skeletons}(ii)}
\label{sec:Zshapes}
We now turn our attention to more complex orthogonal polygons. We show that no
algorithm can beat the competitive ratio of the trivial algorithm, even when
restricting to convex orthogonal 8-gons.  In other words, we prove
\cref{thm:Zgons}.

We focus on convex orthogonal 8-gons which are \emph{Z-shapes}, as defined in
\autoref{sec:prelim}. Typically, a Z-shape has 6 parameters. However, we will
not make use of this flexibilty. In fact, we consider Z-shapes of equal
thickness, which we describe by the four parameters $a,b,w,t$  as illustrated in
\autoref{fig:ZFrameParamsB}.  As a stepping stone we consider \Zframes, i.e.,
Z-shapes of zero thickness, see also \autoref{fig:ZFrameParamsA}.

\IncludeSubfigures{ZFrameParams}{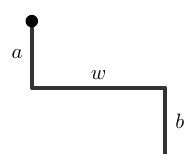}{1,2}{ Illustration of (a) a
    \Zframe with its reference point and (b) a Z-shape. 
}

As a first step, we prove that every online algorithm for \Zskeleton has a
competitive ratio of at least $n$, i.e., \cref{thm:skeletons}(ii). 

\begin{proposition*} [corresponding to \cref{thm:skeletons}(ii)]
	Every online algorithm for \Zskeleton has competitive ratio of at least $n$.
\end{proposition*}
Consider an arbitrary online algorithm \ALG for the online \Zframe packing
problem.  We present a strategy to generate a sequence of \Zframes such that they
can be packed into one bin, but \ALG  uses one bin per \Zframe.  This sequence of
\Zframes is generated by \autoref{alg:ZFrameAdversary}.  After generating a
\Zframe, it reads back the coordinates $(x, y)$ of the reference point chosen by
\ALG.  Due to the fact that $a + b = 1 $ for all generated \Zframes, all valid
placements fulfill $y = 1$ and we may therefore focus on the $x$-coordinates chosen
by \ALG.

\newcommand{\bLow}{b}
\newcommand{\bHigh}{B}
\newcommand{\bIdx}[1]{b_{#1}}
\newcommand{\aIdx}[1]{a_{#1}}

\begin{algorithm}[htb]

  \SetKwFunction{GenZ}{GenerateZSkeletons}

    \Fn(){\GenZ{$n$, \ALG}}{
        $w_0 \gets \nicefrac 12$\;
        $\bLow \gets 0$\;
        $\bHigh \gets 1$\;
        \ForEach{$i \in [n]$}{
            Present $Z_i$ with $w_i := 1 - 2^{- i - 1}$, $\bIdx i := ({\bLow + \bHigh})/{2}$ and $\aIdx i := 1 - \bIdx i$.\;
            $x_i \gets $ $x$-coordinate of $Z_i$ assigned by \ALG\;
            \eIf{$x_i \leq ({1 - w_i})/{2}$} {
                $\bHigh \gets \bIdx i$\;
            } {
                $\bLow \gets \bIdx i$\;
            }
        }

    }
\caption{Strategy to generate  a sequence of \Zframes for a given algorithm \alg.}
\label{alg:ZFrameAdversary}
\end{algorithm}

We start with useful properties of the sequence of \Zframes generated by
\autoref{alg:ZFrameAdversary}.
\begin{lemma}
    \label{lem:properties}
    For the the sequence of \Zframes $Z_1,Z_2,\dots,Z_n$ generated by
    \autoref{alg:ZFrameAdversary}, the following properties hold:
    \begin{romanenumerate}
        \item \label{item:2}Let $I_i$ denote the interval $(\bLow, \bHigh)$
            before generating $Z_i$. Then $I_i \subset I_{i - 1}$, and $|I_i| =
            \frac12 |I_{i - 1}|$.
        \item \label{item:3}$w_j>w_i$ for all $j>i$
        \item \label{item:4} If $x_i \leq (1 - w_i)/2$, then $\bIdx j < \bIdx i$
            for all $j \in [n], j > i$.
        \item \label{item:5}If $x_i > (1 - w_i)/2$, then $\bIdx j > \bIdx i$ for
            all $j \in [n], j > i$.\qedhere
    \end{romanenumerate}
    \label{lemma:ZFrameProperties}
\end{lemma}
\begin{proof}
    We show each property individually.
    \begin{romanenumerate}
        \item  $\bIdx i$ is defined as the midpoint of the interval $(\bLow,
            \bHigh)$, and it becomes a new endpoint.
        
        \item Follows immediately form the definition of $w_i$.
        
        \item As in this case, $\bHigh$ becomes $\bIdx i$, we know that $I_{i +
            1} = (\bLow, \bIdx i)$, so all values in $I_{i + 1}$ are smaller
            than $\bIdx i$. By (\ref{item:2}), $I_j \subseteq I_{i + 1}$, and
            $\bIdx j$ is the midpoint of $I_j$.
        
        \item Analogous to (\ref{item:4}).\qedhere
    \end{romanenumerate}
\end{proof}

Next, we characterize when two \Zframes intersect.
\begin{lemma}
    Consider two \Zframes $Z, Z'$ with $w, w' > \nicefrac12$, $a + b = a' + b' =
    1$ and $b < b'$ which are  placed into the same unit bin at $x$-coordinates
    $x$ and $x'$, respectively. Then  $Z, Z'$ intersect if and only if $x' \leq
    x$ or $x' + w' \leq x + w$.
    \label{lemma:ZFrameIntersection}
\end{lemma}
\IncludeSubfigures{ZLIntersect}{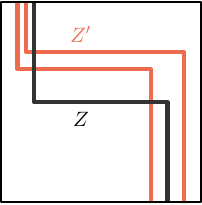}{1,2,3}[0.25,0.25,0.3]{
    Illustrations for \autoref{lemma:ZFrameIntersection}. Two \Zframes $Z$ and
    $Z'$ intersect if and only if the conditions of (a) or (b) are met. 
}[$x'\leq x$, $x' + w' \leq x + w$, $x' > x$ and $x' + w' > x + w$]
\begin{proof}
    The $y$-coordinates of $Z$ and $Z'$ are fixed, as $a + b = a' + b' = 1$.
    Because $b < b'$, the base segments of $Z$ and $Z'$ do not intersect. Also
    by $b < b'$, the lower arm of $Z$ cannot intersect $Z'$, and the upper arm
    of $Z'$ cannot intersect $Z$.  Hence, $Z$ and $Z'$ intersect if and only if
    the upper arm of $Z$ intersects the base of $Z'$, as in
    \autoref{fig:ZLIntersectA}, or the lower arm of $Z'$ intersects the base of
    $Z$, as in \autoref{fig:ZLIntersectB}. This is the case when  $x' \leq x
    \leq x' + w'$ or $x \leq x' + w' \leq x + w$, respectively.  Note that
    $x\leq x' + w'$ always holds because $w,w'>\nicefrac{1}{2}$ and thus
    $x,x'\leq \nicefrac{1}{2}$.  Therefore, the above statement simplifies to
    the fact that $Z$ and $Z'$ intersect if and only if $x' \leq x$ or $x' + w'
    \leq x + w$, as desired.
\end{proof}

\begin{restatable}{lemma}{ZFrameConflict}\label{lemma:ZFrameConflict}
  For every online algorithm \ALG of \Zskeleton,
  \autoref{alg:ZFrameAdversary} generates a sequence of $n$
  \Zframes such that \ALG uses $n$ bins.
\end{restatable}

\begin{proof}
    Consider two \Zframes $Z_i$ and $Z_j$ with $i<j$. We distinguish two cases.
    If $x_i \leq (1 - w_i)/{2}$, we have $\bIdx j < \bIdx i$ by
    \Cref{lem:properties}(\ref{item:4}).  Moreover, \[x_i + w_i \leq (1 + w_i)/2
    = 1 - 2^{-i-2}  = w_{i + 1}\leq w_j\leq x_j + w_j\] where $w_{i+1} \leq
    w_{j}$ by \Cref{lem:properties}(\ref{item:3}), and $0\leq x_j$. Thus, by
    \autoref{lemma:ZFrameIntersection}, they intersect, see also
    \autoref{fig:ZFrameNoFitA}.
    
    \IncludeSubfigures{ZFrameNoFit}{figures/zframes/ZFrameNoFit}{1,2,3}[0.3,0.3,0.3]{
        Illustration for the proof of \autoref{lemma:ZFrameConflict}. {
            Two \Zframes are forced to intersect if (a) $Z_i$
            was placed such that $x_i \leq (1 - w_i)/2$, and $Z_j$ such that $x_j \leq x_i$, or if (b) $x_i > (1 - w_i)/2$
            and $x_j \leq x_i$. (c) No other position for $Z_j$ results in a valid packing.
        }
    }
    
   If $x_i > (1 - w_i)/2$, we have $\bIdx j > \bIdx i$ by
   \Cref{lem:properties}(\ref{item:5}).  By construction and
   \Cref{lem:properties}(\ref{item:3}), we have $w_j \geq w_{i + 1} = 1-2^{-i-2}
   = 1 - (1 - w_i)/2 > 1 - x_i$. Then either $x_j \leq x_i$, implying an
   intersection by \autoref{lemma:ZFrameIntersection}, as in
   \autoref{fig:ZFrameNoFitB}, or $x_j > x_i$, which implies $x_j + w_j > x_i +
   1 - x_i = 1$, a contradiction. This is shown in \autoref{fig:ZFrameNoFitC}.

    Consequently, each newly presented \Zframe $Z_j$ intersects all already
    packed   \Zframes, so the algorithm \alg needs to open a new bin for $Z_j$.
\end{proof}

It remains to show that there exists an offline packing for the generated
\Zframes into one bin. To this end, we show that the set of \Zframes may be
partitioned into two sets $A, B$, such that in $B$, the parameter $b$ grows
monotonically as $w$ grows, while in $A$, the parameter $b$ decreases
monotonically as $w$ grows.

\begin{lemma}
    Let $Z_1, \ldots Z_i$ be the first $i \in \mathbb N$ \Zframes generated by
     \autoref{alg:ZFrameAdversary},    $A_i = \{Z_j \where j
    \in [i - 1], \bIdx j > \bIdx i\}$ and $B_i = \{Z_j \where j \in [i - 1],
    \bIdx j < \bIdx i \}$.  Then, the following condition holds: for each $j, k
    \in [i - 1]$, if $Z_j, Z_k \in B_i$ and $\bIdx j < \bIdx k$ we have $w_j <
    w_k$. Otherwise, if $Z_j, Z_k \in A_i$ and $\bIdx j < \bIdx k$ we have $w_j > w_k$.
    \label{lemma:OrderedSubdivision}
\end{lemma}
\begin{proof}
    We prove this by induction on $i$. For $A_1 = B_1 = \emptyset$, the
    condition is trivially met. For the induction step, we assume that the claim
    holds for $A_{i - 1}$ and $B_{i - 1}$.  If $\bIdx i < \bIdx {i - 1}$, we
    have $A_i = A_{i - 1} \cup \{ Z_{i - 1} \}$ and $B_i = B_{i - 1}$.  The
    condition holds for all pairs of $B_i$ and many pairs of $A_i$ by induction
    hypothesis. It remains to consider the pairs consisting of $Z_{i-1}$ and
    some $Z_k\in A_{i-1}$. Then, because $Z_k \in A_{i - 1}$, we have $\bIdx
    k>\bIdx {i-1}$. Furthermore, $k<i-1$ implies $w_k<w_{i-1}$ by
    \cref{{lem:properties}}(\ref{item:3}), so the condition also holds for these
    pairs.

    The case for $\bIdx i > \bIdx {i - 1}$ is symmetric.
\end{proof}

In order to expand the \Zframes to Z-shapes later, we pack the \Zframes with
some distance to each other.  In particular, we present a packing with a lower
bound on the \emph{horizontal gaps}, i.e., the minimum horizontal distance
between adjacent \Zframes. In the following, we denote the parameters of an
\Lshape that we use as a container by the tuple~$(\ell_x,w_x,\ell_y,w_y)$.

\begin{lemma}
    \label{lemma:ZLShapePacking}
    Let $\epsilon > 0$.  Consider a set of \Zframes $Z_1, \ldots, Z_k$ with
    $\bIdx i + a_i = 1$ for all $i \in [k]$ and $w_1 < \ldots < w_k$ and $\bIdx
    1 < \ldots < \bIdx k$. Then  $Z_1, \ldots, Z_k$ can be packed  into an
    \Lshape with parameters $(\epsilon+w_k,\epsilon,1,\bIdx k)$ such that all
    horizontal gaps have size at least $\nicefrac\epsilon k$.
    \end{lemma}
\begin{proof}
    
    \providecommand{\nicefrac}[2]{\frac{#1}{#2}}
    We prove this by induction on $k$. The induction base for $k = 1$ is clear.
    For the induction step, set $\epsilon' := \epsilon (1-\nicefrac{1}{k})$. By
    induction hypothesis,  $Z_1, \ldots, Z_{k - 1}$ can be packed into the
    L-shape with parameters $(\epsilon'+w_{k-1},\epsilon',1,\bIdx {k-1})$ with
    horizontal gaps of size $\nicefrac{\epsilon'}{k - 1} = \nicefrac{\epsilon}{k}$. For
    an illustration consider \autoref{fig:ZLRegionB}.

    \IncludeSubfigures{ZLRegion}{figures/zframes/zframes3.pdf}{4,5,6,7}[0.23,0.23,0.22,0.29]{
        Illustrations for \Cref{lemma:ZLShapePacking,lemma:ZLShapePackingInverse,lemma:OptimalZFrames}.
        A set of \Zframes whose widths grow with the lower arm lengths
can be packed into the \Lshape in \ref{fig:ZLRegionA}, by the inductive
        argument visualized in \ref{fig:ZLRegionB}. If the widths monotonically
        decrease as the lower arm lengths grows, they  pack into a rotated \Lshape
        in \ref{fig:ZLRegionC}. Finally, these shapes can be combined as in
    \ref{fig:ZLRegionD}. 
    }
    
    As $\bIdx {k - 1}<\bIdx k$, and $w_{k - 1}<w_k$, we can place $Z_k$ next to
    the \Lshape with a gap of $\nicefrac{\epsilon}{k}$ at the upper arm. It is
    easy to check that the horizontal gaps have size $\geq
    \nicefrac{\epsilon}{k}$ as $w_{k - 1}<w_k$. Because
    $\epsilon'+\nicefrac{\epsilon}{k}=\epsilon$, the parameters of the resulting
    \Lshape are as desired.
\end{proof}

By rotation, \autoref{lemma:ZLShapePacking} also implies the following.
\begin{lemma}
    Let $\epsilon > 0$.  Consider a set of \Zframes $Z_1, \ldots, Z_k$ with
    $\bIdx i + a_i = 1$ for all $i \in [k]$ and $w_1 < \ldots < w_k$ and $\bIdx
    1 > \ldots > \bIdx k$. Then  $Z_1, \ldots, Z_k$ can be packed  into a
    $\pi$-rotated \Lshape  with parameters $(\epsilon+w_k,\epsilon,1,a_k)$ (see
    \autoref{fig:ZLRegionC}), such that all horizontal gaps have size at least
    $\nicefrac\epsilon k$.
    \label{lemma:ZLShapePackingInverse}
\end{lemma}

Using \cref{lemma:OrderedSubdivision,lemma:ZLShapePacking,lemma:ZLShapePackingInverse}, we show that all generated \Zframes can be packed into a bin.
\begin{restatable}{lemma}{OptimalZFrames}\label{lemma:OptimalZFrames}
  Let $Z_1, \ldots, Z_n$ be \Zframes generated by some execution of
  \autoref{alg:ZFrameAdversary}.  Then, $Z_1, \ldots, Z_n$ can
  be packed into
  one bin, such that all horizontal gaps between the \Zframes and the bin
  boundary have width at least $\nicefrac1n\cdot 2^{-n - 3}$.
\end{restatable}
\begin{proof}
    By \autoref{lemma:OrderedSubdivision}, there are sets $A:=A_n$ and $B:=B_n$
    such that  $A\cup B \cup\{Z_n\}=\{Z_1, \ldots, Z_n\}$ and the widths are
    monotonically increasing with the base height in $B$ and decreasing in $A$.
    We choose $\epsilon := \nicefrac1n\cdot 2^{-n - 3}$ and apply
    \autoref{lemma:ZLShapePacking} to $B$ and
    \autoref{lemma:ZLShapePackingInverse} to  $A$. This produces two packings in
    \Lshape regions, which can be placed next to each other as depicted in
    \autoref{fig:ZLRegionD}.  We argue that the packing is valid: We know $\bIdx
    j < \bIdx n < \bIdx k$ for all $j, k \in [n]$ with $Z_j \in B$ and $Z_k \in
    A$. With horizontal gaps of size $\nicefrac \epsilon 2$, the total width is
    $ 4\epsilon + w_n =\nicefrac1n\cdot 2^{-n - 1}+(1-2^{n-1})= 1$, and the
    height is precisely $1$, so it fits into a unit bin.

    Within the \Lshape regions, the horizontal gaps are already large enough.
    The gaps between the regions and $Z_n$ have size $\nicefrac{\epsilon}{2} \geq
    \nicefrac{\epsilon}{n}$ for $n \geq 2$.  For $n = 1$, there are no
    horizontal gaps.  Finally, by construction, the horizontal gaps between any
    \Zframe and the bin boundary is at least $\nicefrac{\epsilon}{2} \geq
    \nicefrac{\epsilon}{n}$ for $n \geq 2$, and at least $2\epsilon \geq
    \nicefrac{\epsilon}{n}$ for $n = 1$.
\end{proof}

\Cref{lemma:OptimalZFrames,lemma:ZFrameConflict} directly imply
\Cref{thm:skeletons}(ii). In order to prove \Cref{thm:Zgons}, we slightly modify
\cref{alg:ZFrameAdversary}. Note that we already guaranteed that the \Zframes we
generate can not only be placed into a common bin, but that such a placement can
be made such that at least $\nicefrac1n\cdot 2^{-n - 3}$ space remains between
each pair of adjacent frames, and between each frame and the bin boundary.
Furthermore, making the shapes thicker never helps the packing algorithm.

\Zgons*
\begin{proof}
    Consider an online algorithm \ALG.  The idea is to use a slightly modified
    version of \cref{alg:ZFrameAdversary}, where we obtain Z-shapes by increasing the thickness from
    $0$ to a positive value: \Cref{lemma:OptimalZFrames}
    ensures that the \Zframes can be packed into a unit bin such that they have
    horizontal gaps of at least  $ \nicefrac1n\cdot 2^{-n - 3}$.  By
    construction, the vertical distance between any two base segments is at
    least $2^{-n}$.  Therefore, it is possible to replace the \Zframes with
    Z-shapes that are defined by the same parameters and  a thickness of
    $\nicefrac 1n\cdot 2^{-n-3}$. In particular,  all Z-shapes  fit into a single bin
    if packed optimally. Moreover, the argument of \Cref{lemma:ZFrameConflict} still
    guarantees that \ALG uses $n$ bins. Thus, \ALG has a competitive
    ratio of $n$.
\end{proof}

\section{Conclusion and open problems}
In this work, we investigated various packing problems of orthogonal polygons
under translation. For \Lshapes, the yet smallest open case, we wondered whether
they behave more like rectangles or convex polygons. For bin packing, we
provided a surprisingly large lower bound on the competitive ratio of about
$\nicefrac{n}{\log n}$ for general \Lshapes and $n$  for general orthogonal
polygons. That means that the trivial algorithm, which packs each item into its
private bin, is worst-case optimal for orthogonal polygons. These lower bounds
are much higher than for convex polygons, cf.~\cref{tab:2}.  In contrast, for
symmetric and small \Lshapes there exist constant competitive algorithms. This
insight implies that in terms of the competitive ratio, \Lshapes behave more
like rectangles for perimeter packing, area packing, and the online critical
density.

Many interesting problems remain for future work. Most notably, it remains to
improve the bounds or even close the gaps for the problems stated in
\cref{tab:2,table:BPresults}. This includes the problem of determining the best
competitive ratio for online bin packing of L-shapes, as it remains open whether
it it possible to beat the trivial algorithm for \trans.

\bibliography{bib}

\end{document}